\documentclass[journal]{IEEEtran}

\usepackage[dvipsnames,svgnames,x11names,table]{xcolor}
\usepackage{soul,framed}
\usepackage{amsfonts}
\usepackage{multicol}
\usepackage{multirow}
\usepackage{booktabs}
\usepackage{bbm}
\usepackage{mathrsfs}
\colorlet{shadecolor}{yellow}
\usepackage[pdftex]{graphicx}
\graphicspath{{../pdf/}{../jpeg/}}
\DeclareGraphicsExtensions{.pdf,.jpeg,.png}
\usepackage[numbers,sort&compress]{natbib}
\usepackage[cmex10]{amsmath}
\usepackage{array}
\usepackage{mdwmath}
\usepackage{subfigure}
\usepackage{graphicx}
\usepackage{tabularx}
\usepackage{amsmath}
\usepackage{lscape}
\usepackage{mdwtab}
\usepackage{eqparbox}
\usepackage{float}
\usepackage{bm}
\usepackage{lineno}
\usepackage{makecell}
\usepackage{tikz}
\usetikzlibrary{positioning}
\usepackage{calc}
\usepackage{diagbox}
\usepackage{threeparttable}
\usepackage{verbatim}
\usepackage{algorithm}
\usepackage{algorithmic}
\usepackage[figuresright]{rotating}
\usepackage[breaklinks=true]{hyperref}
\usepackage{booktabs}
\usepackage{graphicx}
\usepackage{colortbl}
\usepackage{placeins}
\usepackage{dblfloatfix}
\usepackage{mdframed}
\usepackage{tcolorbox}
\usepackage{xcolor}
\tcbset{
    myquote/.style={
        colback=blue!5!white,
        colframe=blue!75!black,
        fonttitle=\bfseries,
        boxrule=0.3mm,
        arc=2mm,
        left=1mm,
        right=1mm,
        top=1mm,
        bottom=1mm,
        enhanced,
        sharp corners,
        boxsep=1mm,
    }
}

\usepackage{CJKutf8}
\usepackage{amsthm}
\usepackage{amsthm}
\newtheorem{lemma}{Lemma}
\def\BibTeX{{\rm B\kern-.05em{\sc i\kern-.025em b}\kern-.08em
    T\kern-.1667em\lower.7ex\hbox{E}\kern-.125emX}}

\title{
A Plug-and-Play Method for Improving Imperceptibility and Capacity in Practical Generative Text Steganography
}
 \author{Kaiyi Pang}

\begin{document}
\begin{CJK*}{UTF8}{gbsn}
\maketitle

\begin{abstract}

Linguistic steganography embeds secret information into seemingly innocuous text to safeguard privacy under surveillance. Generative Linguistic Steganography (GLS) leverages the probability distributions of Language Models (LMs) and applies steganographic algorithms during generation. GLS has attracted increasing attention with the rise of Large Language Models (LLMs). To strengthen security, prior work has focused on distribution-preserving steganographic algorithms that minimize the gap between stego sampling and random sampling from the model. However, their reliance on model distributions, which often deviate from real-world cover texts, leads to limited imperceptibility when facing steganalysis detectors in practical settings. Moreover, LLM distributions tend to be more deterministic, reducing entropy and thus lowering embedding capacity.
In this paper, we propose a plug-and-play method that reconstructs the distributions of language models used for generative linguistic steganography. FreStega dynamically adjusts token probabilities from the language model at each step of autoregressive stego text generation, leveraging both sequential and spatial dimensions. Extensive experiments on four LLMs, three benchmark datasets, and four distribution-preserving steganographic baselines demonstrate that, by reforming the distribution, FreStega improves the imperceptibility of stego text in realistic scenarios and increases steganographic capacity by 15.41\%, without degrading the quality of the generated stego text.

\end{abstract}


\section{Introduction}
Steganography embeds secret messages within seemingly innocuous carriers, enabling private communication while reducing the risk of detection under pervasive surveillance.
Among the various carriers used in the information age, including images \cite{yang2023provably} and audio \cite{gopalan2003audio}, text plays an important role in most Internet scenarios because of its simplicity, versatility, and robustness. Consequently, text steganography has attracted considerable attention \cite{zhang2021provably, meteor2021, ding2023discop, ziegler2019neural}.

With the rapid advancement of language models, generative linguistic steganography has emerged as a mainstream paradigm. As illustrated in Figure \ref{fig:general framework}, it approximates the steganographic channel using language models and embeds secret messages by controlling token selection during autoregressive text generation.
State-of-the-art distribution-preserving steganography algorithms \cite{meteor2021, ding2023discop} ensure that tokens generated through random sampling from the model are computationally indistinguishable from those generated via steganography.

\begin{figure}
    \centering
    \includegraphics[width=1\linewidth]{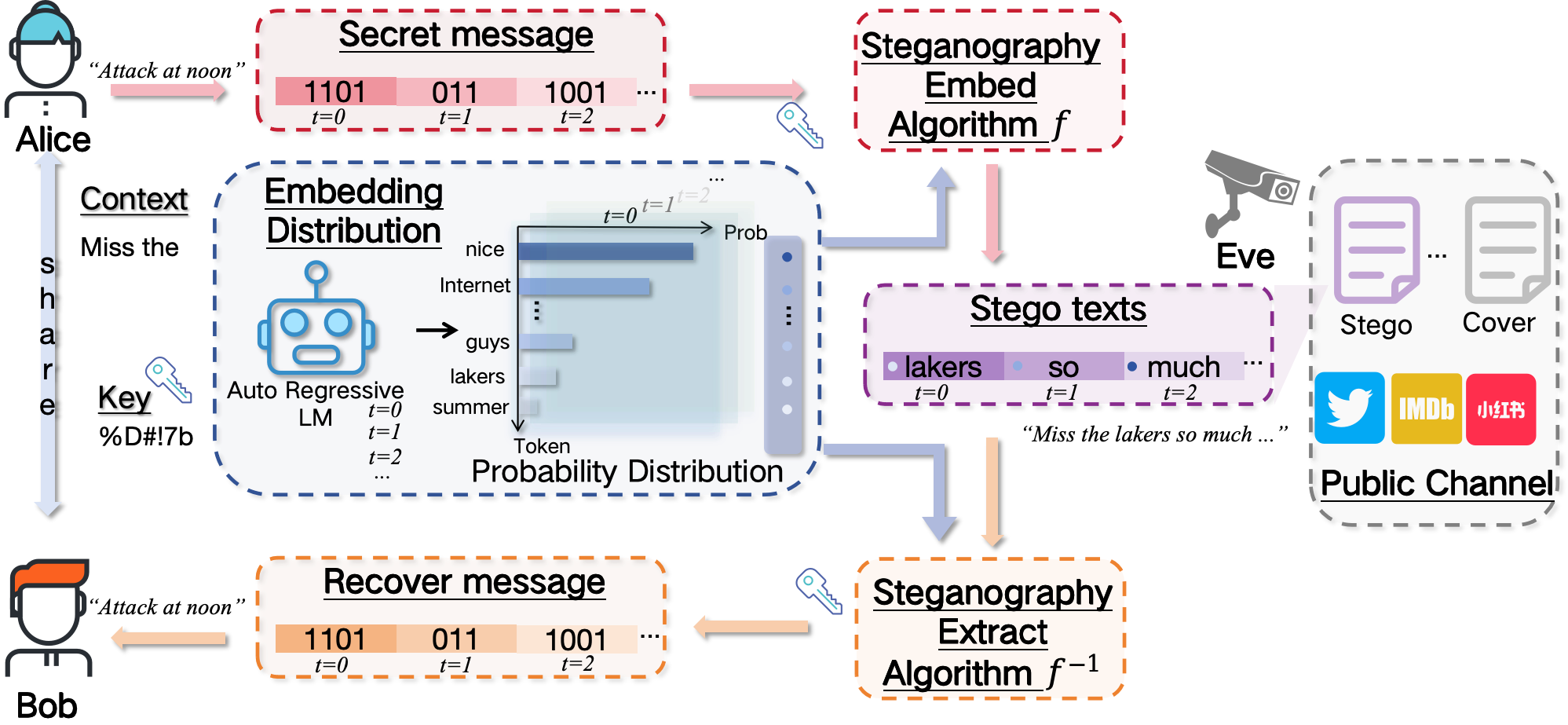}
    \caption{Generative Linguistic Steganography Framework: Alice and Bob covertly communicate over a public channel monitored by Eve. Alice embeds secret information during the language model's token selection process to generate stego text. Using a shared key, language model, and prompt, Bob then recovers the secret information upon receiving the stego text.}
    \label{fig:general framework}
\end{figure}

\begin{figure*}[!t]
    \vspace{-0.8em}
    \centering
    \includegraphics[width=1.0\linewidth]{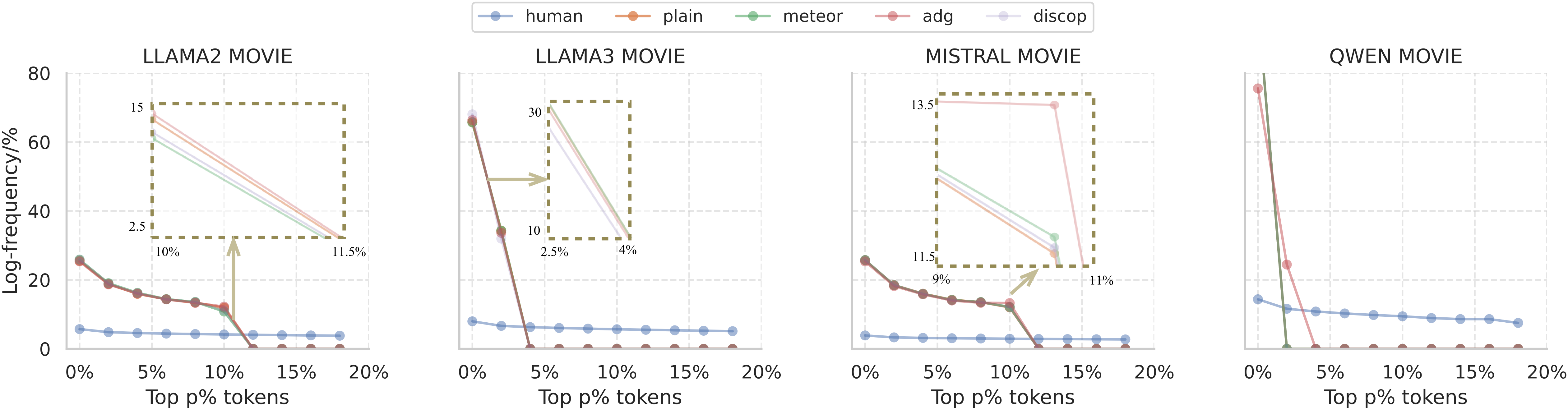}
    \caption{
Token frequency distribution comparison among cover text, random sampling (plain), and stego texts generated by METEOR \cite{meteor2021}, ADG \cite{zhang2021provably}, and DISCOP \cite{ding2023discop} on IMDB \cite{IMDB}. Each model uses its own tokenizer, so the tokenized cover distribution may vary slightly across models.}
    \label{frequency}
    \vspace{-0.6em}
\end{figure*}
\begin{figure}[!t]
    \vspace{-0.8em}
    \centering
    \includegraphics[width=1\linewidth]{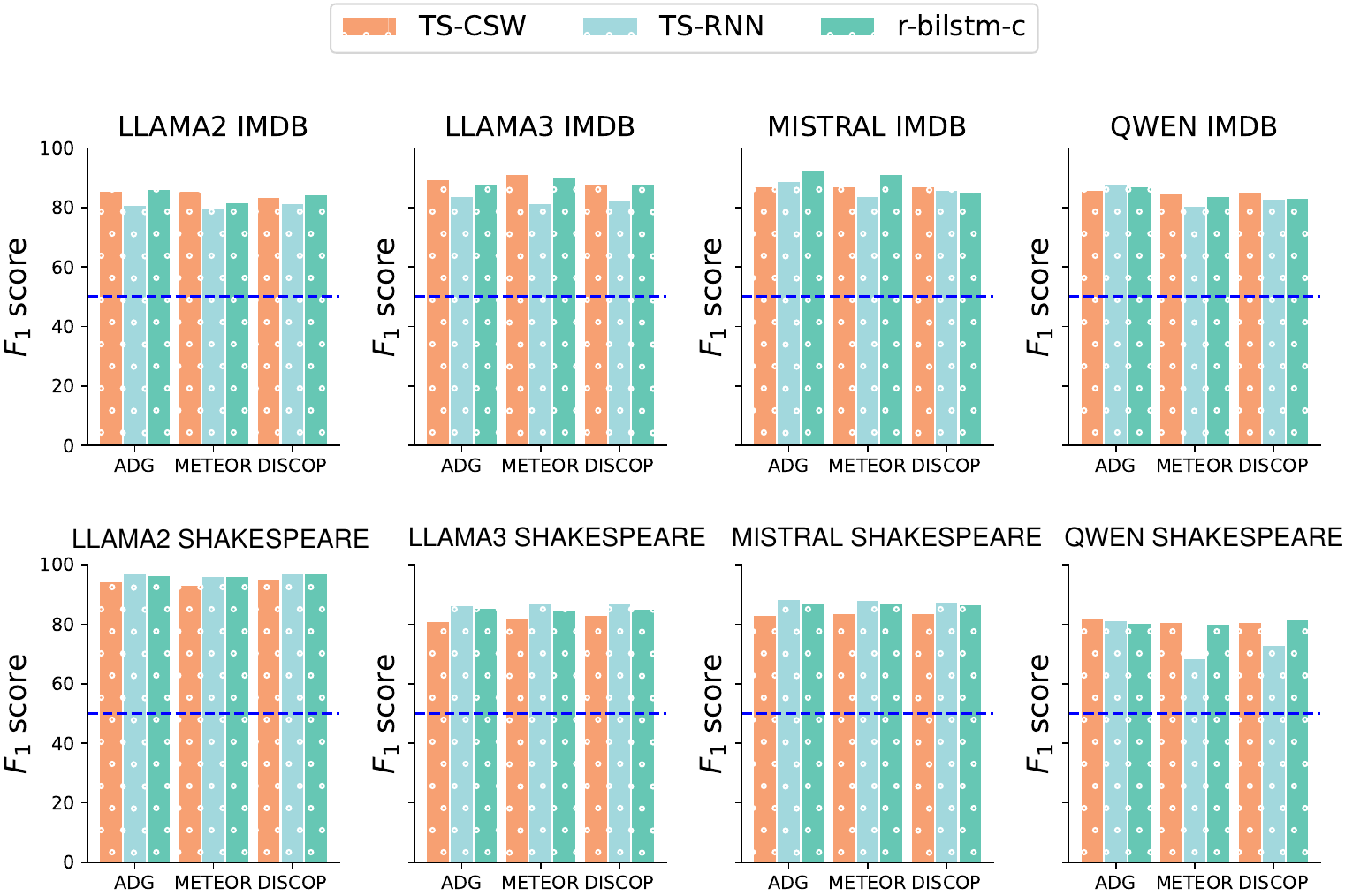}
    \caption{\textcolor{black}{F1 classification scores of state-of-the-art steganography methods (ADG \cite{zhang2021provably}, METEOR \cite{meteor2021}, DISCOP \cite{ding2023discop}) against classic steganalysis classifiers (TS-CSW \cite{TS-CSW}, TS-RNN \cite{TS-RNN}, and R-BiLSTM-C \cite{r-bilstm-c}) on the IMDB \cite{IMDB} and SHAKESPEARE \cite{shakespere} in real-world scenarios.}}
    \label{steganalysis}
    \vspace{-0.6em}
\end{figure}

Despite significant advancements, generative linguistic steganography methods still face practical challenges in the era of large language models (LLMs), notably the \textbf{de facto imperceptibility challenge} and the \textbf{low-capacity challenge}.
Despite the fluency of texts generated by large language models, even the most advanced generative steganography schemes produce texts that can be detected and blocked by automated detectors in real-world applications like social networks, often resulting in account suspensions—even before considering more advanced steganalysis techniques.
\textbf{\textcolor{black}{Ensuring stego text is indistinguishable from randomly sampled language model text is not sufficient for real-world applications.}}
In real-world scenarios, censors can access extensive real-world human data from the domain as cover samples to train stego detectors, a common setup in existing steganalysis tasks \cite{yang2022tifs,wang2023linguistic,xue2023adaptive}.
To verify this, we tested several distribution-preserving steganography algorithms (ADG \cite{zhang2021provably}, METEOR \cite{meteor2021}, DISCOP \cite{ding2023discop}) with popular language models like \textsc{Qwen2} \cite{qwen}, \textsc{Mistral v0.3} \cite{mistral}, \textsc{Llama2}, and \textsc{Llama3} \cite{llama2,llama3}. As shown in Figures \ref{frequency} and \ref{steganalysis}, stego texts generated by state-of-the-art distribution-preserving algorithms \cite{meteor2021,ding2023discop} exhibit discrepancies from real-world cover texts (e.g., real comments on IMDB) in both word frequency distribution and steganalysis performance. The steganalysis follows the mainstream paradigm \cite{yang2022tifs,wang2023linguistic,xue2022effective,xue2023adaptive,TS-CSW}, focusing on distinguishing real cover texts in the target domain from stego texts. This occurs because existing language model distributions do not match cover texts well, and distribution-preserving methods rely entirely on these misaligned distributions. Biased language model estimations lead to a noticeable divergence between stego texts and real cover texts.

Moreover, the embedding capacity of existing generative steganography methods is constrained by the inherent distribution of the language model. Since LLMs are currently the best estimators of text channels, generative text steganography predominantly utilizes distributions generated by LLMs for steganographic sampling.
LLMs tend to produce more deterministic responses compared to smaller models \cite{lowentropy,liao2024co}, as RLHF training often reduces the entropy of their distributions. This decrease in entropy consequently limits their embedding capacity.

In summary, existing distribution-preserving algorithms are heavily dependent on the language model's distribution. However, current language model distributions still struggle to meet the practical demands of steganography in terms of imperceptibility and capacity.
To alleviate the imperceptibility and capacity limitations of generative linguistic steganography based on imperfect language-model distributions, we propose FreStega, a plug-and-play distribution reformation method.
We reform the model's distribution by dynamically adjusting the sharpness of the prediction distribution at each time step (we call this \textbf{Sequential Adjustment}) and by adjusting the probability of each token using real cover text as guidance (we call this \textbf{Spatial Adjustment}). We conducted a preliminary exploration of certain methods within spatial adjustment \cite{pang2024fremax} in our previous work. Building on that foundation, we further optimized the spatial adjustment method and introduced sequential adjustment.
Sequential Adjustment modifies the temperature based on instantaneous entropy, reducing distribution sharpening and enhancing the diversity and embedding capacity of the stego text. Spatial Adjustment boosts the probabilities of tokens common in the target domain while suppressing overconfidence in the language model, enabling quick adaptation to target-domain distributions and enhancing imperceptibility.

Our main contributions are summarized as follows:

\begin{itemize}
    \item \textbf{Plug-and-play and Model-Agnostic Design}:
    FreStega can be seamlessly integrated with existing generative linguistic steganography algorithms in a plug-and-play manner, including grouping-based methods such as ADG \cite{zhang2021provably} and DISCOP \cite{ding2023discop}, as well as AC-based methods \cite{ziegler2019neural,shen2020near}. \textcolor{black}{FreStega operates during token decoding, effectively aligning stego text generated by an autoregressive language model with the target domain without requiring any additional training.}
    \item \textbf{Dual-dimensional Dynamic Distribution Adjustment}:
 FreStega dynamically adjusts the probability distribution of language models in both sequential and spatial dimensions. This effectively approximates the distribution of the target domain, mitigating the de facto imperceptibility challenge of existing steganography methods caused by their heavy reliance on language model distributions when facing steganalysis in real-world application settings.

    \item \textbf{Effectiveness and Flexibility of Alignment}:
Extensive experiments show that FreStega reduces distribution divergence between stego and cover texts, increasing imperceptibility against steganalysis detectors in real-world scenarios while enhancing capacity without compromising linguistic quality. FreStega is also highly flexible, requiring only about 100 samples to achieve effective domain alignment in Spatial Adjustment. Even without access to target-domain text, Sequential Adjustment can still improve the quality of stego text.

\end{itemize}

\section{Method}

\subsection{Narrowing the Gap Toward Real-World Scenarios}
When applying generative text steganography schemes in real-world scenarios, a fundamental consideration is that the generated stego texts should seamlessly fit into the environment.
Specifically, the stego sampling distribution \( p_S \) should closely resemble the target-domain cover distribution $p_{\mathcal{E}}$:
$
p_S = \arg\min_{p_S} D_{KL}(p_{\mathcal{E} }\| p_S).
$
We denote by \( p_{\mathcal{M}} \) the original language model distribution, which represents the random sampling process of the language model \(\mathcal{M}\). The embedding process introduces noise \( \epsilon \), leading to a new distribution \( p_S = p_{\mathcal{M}}+ \epsilon \),
The divergence from the target-domain distribution can be decomposed exactly as:
\begin{equation}
\begin{split}
D_{KL}(p_{\mathcal{E}} \| p_S) & = D_{KL}(p_{\mathcal{E}} \| p_M + \epsilon) \\
& = D_{KL}(p_{\mathcal{E}} \| p_M) \\
&\quad -\sum_x p_{\mathcal{E}}(x)
\log\left(1+\frac{\epsilon(x)}{p_M(x)}\right).
\end{split}
\end{equation}
For an ideal statistically secure distribution-preserving steganography algorithm, \(\epsilon = 0\) is considered theoretically achievable, meaning that \( D_{KL}(p_S \| p_M) = 0 \), and the quality of the stego text is primarily determined by \( p_M \).
However, reducing $ D_{KL}(p_S \| p_M)$ alone is not sufficient for higher imperceptibility in real-world scenarios.
This also highlights the issue of heavy dependence on the language model distribution discussed earlier. The security and quality of the stego text largely depend on $D_{KL}(p_{\mathcal{E}} \| p_M)$.
Narrowing the gap between the language model distribution \( p_M \) and the target-domain text distribution \( p_{\mathcal{E}} \) improves model--environment alignment for both distribution-preserving and non-distribution-preserving algorithms. For distribution-preserving methods, the stego sampling distribution closely follows the input model distribution, so improving $p_M$ directly improves $p_S$. For non-distribution-preserving methods, the encoding procedure introduces additional sampling noise, but a better-aligned base distribution still reduces the model--environment component $D_{KL}(p_{\mathcal{E}} \| p_M)$.

\textit{KL Divergence of the Non-Distribution-Preserving Algorithm}
\label{kl}
For a non-distribution-preserving algorithm, suppose that for each value \(x\) in the support there exists \(0<\delta(x)\leq 1/3\) such that \(\delta(x)\leq p_{\mathcal{M}}(x),p_{\mathcal{S}}(x)\leq 1-\delta(x)\). We then have:
\begin{equation}
\begin{aligned}
D_{KL}(p_{\mathcal{E}}\|p_{{\mathcal{S}}})
&=\sum_x p_{\mathcal{E}}(x)\left[\log\frac{p_{\mathcal{E}}(x)}{p_{\mathcal{M}}(x)}\right.\\
&\qquad\left.-\log\left(1+\frac{\epsilon(x)}{p_{\mathcal{M}}(x)}\right)\right].
\end{aligned}
\end{equation}
The magnitude of the latter term is bounded as in Equation~\ref{math}, with a detailed proof in Lemma~\ref{lemma1}.
\begin{equation}
\begin{aligned}
\left|\log\left(1+\frac{\epsilon (x)}{p_{\mathcal{M}}(x)}\right)\right|
&=\left|\log\left(1+\frac{p_{\mathcal{S}}(x)-p_{\mathcal{M}}(x)}{p_{\mathcal{M}}(x)}\right)\right|\\
&\leq \frac{1-2 \delta(x)}{\delta(x)}.
\end{aligned}
\label{math}
\end{equation}
This decomposition separates the model--environment mismatch from the sampler-induced perturbation.
\begin{equation}
\begin{aligned}
D_{KL}(p_{\mathcal{E}}\|p_{\mathcal{S}})
&\geq D_{KL}(p_{\mathcal{E}}\|p_{\mathcal{M}})\\
&\quad-\sum_{x}\frac{1-2\delta(x)}{\delta(x)}p_{\mathcal{E}}(x).
\end{aligned}
\end{equation}

\begin{lemma}
Let \(0<\delta\leq 1/3\). If \( \delta \leq p, q \leq 1-\delta \), then the following holds:
\[
    \left| \log\left(1 + \frac{q - p}{p}\right) \right| \leq \frac{1 - 2\delta}{\delta}.
\]
By symmetry, the same inequality applies to \( \left| \log\left(1 - \frac{q - p}{1 - p}\right) \right| \).
\label{lemma1}
\end{lemma}

\begin{proof}
Let \(R=(1-\delta)/\delta\). Since \(0<\delta\leq 1/3\), we have \(R\geq 2\). From \(p,q\in[\delta,1-\delta]\), it follows that \(R^{-1}\leq q/p\leq R\), and hence
\[
\left|\log\left(1+\frac{q-p}{p}\right)\right|
=\left|\log\frac{q}{p}\right|
\leq \log R.
\]
For base-2 logarithms, \(\log R\leq R-1\) for \(R\geq2\): indeed, \(h(R)=R-1-\log R\) satisfies \(h(2)=0\) and \(h'(R)=1-1/(R\ln2)>0\) on this interval. Therefore,
\[
\left|\log\frac{q}{p}\right|\leq R-1=\frac{1-2\delta}{\delta}.
\]
Replacing \((p,q)\) with \((1-p,1-q)\), which also lie in \([\delta,1-\delta]\), proves the symmetric inequality.
\end{proof}

There are two main intuitive approaches to generating stego text that align with the characteristics of the target domain \(\mathcal{E}\).

The first approach is to \textbf{fine-tune} or \textbf{refactor} \(\mathcal{M}\) into $\mathcal{M}_\mathcal{E}$ \cite{trainmodel2,trainmodel3} using a large corpus from \(\mathcal{E}\), allowing the model to learn domain-specific characteristics during targeted training. The second approach is to \textbf{optimize prompts}, including prompt-tuning or instruction-guiding methods \cite{prompttuning1,prompttuning2}, which involve finding a suitable prompt or its soft embedding. This is more effective for large models with strong instruction-following capabilities \cite{prompt-direct,prompt-direct2}, enabling them to generate text that reflects the style of \(\mathcal{E}\).
However, because steganography applications operate in complex environments where it is difficult to obtain direct and reliable supervision signals or explicit optimized prompts, to the best of our knowledge there has been no prior work that directly optimizes an LLM for the steganography task.

Our method works during the \textbf{decoding stage} of language models and improves the decoding distribution \(P(x_t|x_{<t},\mathcal{H}_\mathcal{E})\) produced by either fine-tuned or prompt-guided models. \textcolor{black}{ In other words, \textbf{our method is compatible with both fine-tuning and prompt-based approaches.} The compatibility analysis is provided in the supplementary material.}
Most existing controllable text generation methods \cite{pplm,GEDI,FUDGE} at the decoding stage focus on explicit semantic attribute control and often require additional specially designed classifiers \cite{pplm,FUDGE}. \textcolor{black}{To our knowledge, there have been only two prior attempts to apply decoding-phase methods for stego text generation. One, by \cite{chencheng}, builds on PPLM \cite{pplm}, but this approach requires dimensional alignment between the classifier and the language model, making it difficult to scale to current LLMs. Another preliminary attempt \cite{pang2024fremax} adjusts the softmax to infuse a certain target-domain style into the generated stego text. However, this method overlooks capacity issues from the low entropy of LLMs and fails to suppress their inherent AI style.}

\subsection{Overview}
Our objective is twofold: to mitigate the intrinsic bias in language models, thereby producing stego text with greater diversity and a style closer to cover texts in the application domain, and to increase the entropy of generated text while preserving the linguistic quality of generated stego texts.

We adjust the language model distribution \(p_M\) for better
steganographic sampling by addressing two key aspects: aligning it more closely with $p_{\mathcal{E}}$ to enhance practical imperceptibility, and mitigating the sharpening effect to slightly increase entropy while preserving linguistic quality, thereby boosting embedding capacity.

\begin{figure}[htbp]
    \centering
    \includegraphics[width=1\linewidth]{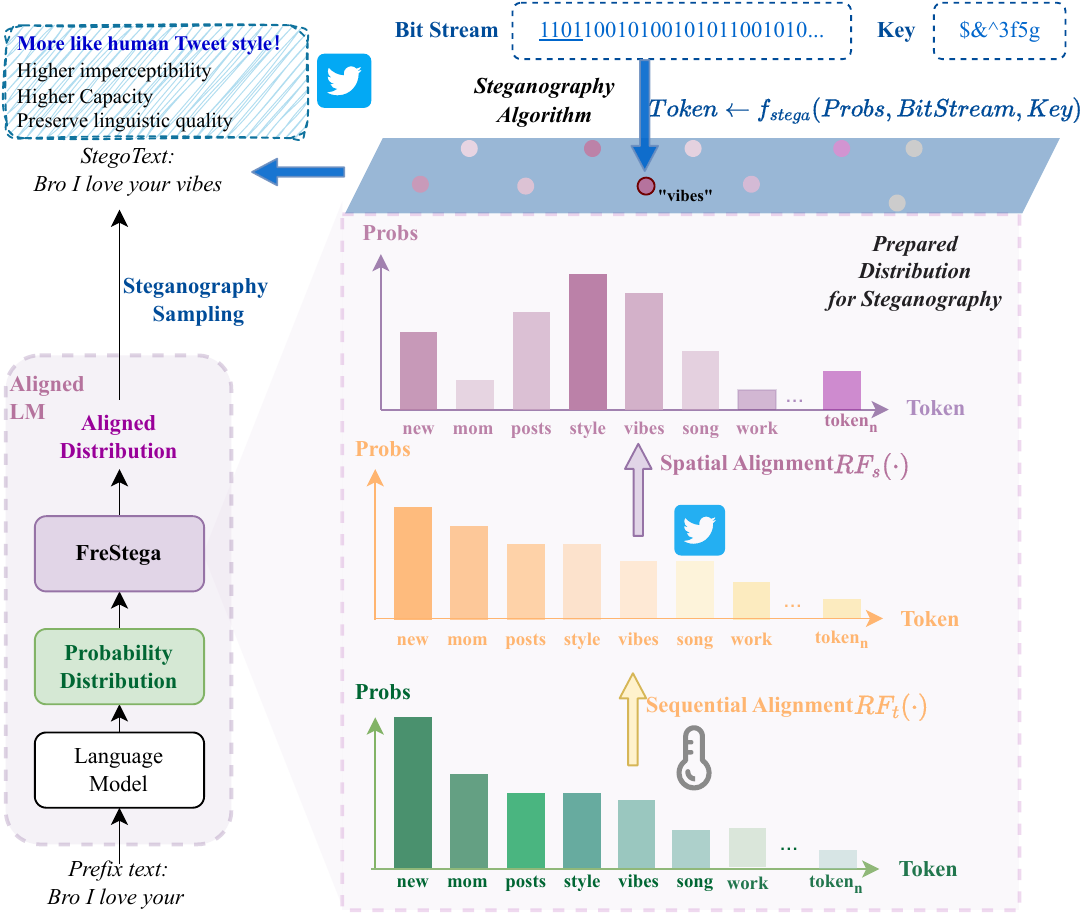}
  \caption{Overview of FreStega: We first adjust the temperature at each time step based on the instantaneous entropy in the sequential dimension. Then, we refine the language model’s distribution in the spatial dimension, aligning it with the target-domain cover distribution.}
    \label{fig:method}
\end{figure}

Consider a text sequence $\textbf{x}=x_{0},x_{1},...,x_{N-1}$, where each $x_{i}=w_k\in \mathscr{V}$ is a token from the model's tokenizer vocabulary.
The autoregressive language model used as the base model for generative steganography models the distribution of a sentence $\textbf{x}$ as a product of conditional probabilities $p(x_i|\textbf{x}_{0:i-1})$ using the chain rule:
$
    p(x)=\prod_{i=0}^{N-1} p(x_i|\textbf{x}_{0:i-1}).
$

The probability of token $w_k$ at the $t$-th position in sequence \textbf{x} is computed by Equation \ref{Eq.5}:
\begin{equation}\label{Eq.5}
    \hat{p}\left(x_{t}=w_{k}|{\bf x}_{0:t-1}\right)=\frac{e^{{\bf h}^{t} \cdot {\bf o}^{t}_{k}}}{\Sigma_{w_{i}\in \mathscr{V}}\:\{e^{{\bf h}^{t} \cdot {\bf o}^{t}_{i}}\}},
\end{equation}
where ${\bf o}^{t}_{k}$ and ${\bf h}^{t}$ denote the output-layer vector and the hidden state of the LM at the $t$-th step in the generation process, respectively.

Our method modifies the conditional probability distribution $p(x_t|\textbf{x}_{<t})$ across both sequential and spatial dimensions to approximate the target-domain cover distribution, as shown in Figure \ref{fig:method}.
FreStega adjusts the original distribution $p_M$ to a distribution better suited for steganographic sampling algorithms.
The steganography algorithm does not require any modifications; it determines the next token based on the key and secret information within the new distribution.

In the sequential dimension, we mildly smooth the model's output to support a more diverse and natural flow of language. By incorporating this dynamic temperature regulation at every time step,
we aim to increase embedding capacity and creativity in the generated stego text while preserving fluency.

In the spatial dimension, we align the language model's distribution with the token frequency distribution of the target cover corpus.
It is worth noting that we use the token frequency distribution of the language model in our experiments because token frequency is the most direct marginal distribution of text. In fact, our method can be applied to other marginal distributions as long as the distribution can be modeled explicitly, such as sentiment or topic.

As a result, the language model is better adapted to generating texts that closely resemble target-domain cover texts. The pseudocode of the overall method is shown in Algorithm \ref{alg1}.
We elaborate on the sequential alignment and spatial alignment in Algorithm \ref{alg1} in Sections~\ref{sequential dimension} and~\ref{spatial dimension}, respectively.

\begin{algorithm}[t]
\caption{Probability distribution reformation and stego text generation}

\renewcommand{\algorithmicrequire}{\textbf{Input:}}
\renewcommand{\algorithmicensure}{\textbf{Output:}}
\label{alg1}
\begin{algorithmic}[1]
\REQUIRE target-domain corpus $\mathscr{D}$; model-generated corpus $\mathscr{M}$; language model $LM$; tokenizer $T$; token vocabulary $\mathscr{V}$; input text ${\bf x}_{0:t-1}$; steganographic algorithm $f$; shared key $\mathcal{K}$; secret bits $B$.

\ENSURE output text ${\bf x}$.
\STATE $\text{$\cdot\cdot\cdot$Initial frequencies}$
    \FOR{$v \in \mathscr{V}$}
    \STATE $D_{freq}(v) \leftarrow \frac{Count^{T(\mathscr{D})}(v)}{\#T(\mathscr{D})}$
    \STATE $M_{freq}(v) \leftarrow \frac{Count^{T(\mathscr{M})}(v)}{\#T(\mathscr{M})}$
    \ENDFOR
    \WHILE{$x_{t-1} \neq \text{EOS}$}
  \STATE $\text{$\cdot\cdot\cdot$Sequential alignment}$
    \STATE $\mathbf{L} \leftarrow LM({\bf x}_{0:t-1})$
    \STATE $\mathbf{P} \leftarrow \frac{e^{\mathbf{L}}}{\text{sum}(e^{\mathbf{L}})}$
    \STATE $E_{t-1} \leftarrow -\mathbf{P^{T}}\cdot \log(\mathbf{P})$

        \STATE $TP_{t-1} \gets Temp(E_{t-1}) $
        \STATE $\mathbf{L_{RF}} \leftarrow RF_t(\mathbf{L},TP_{t-1})$
        \STATE $\text{$\cdot\cdot\cdot$ Spatial alignment}$
        \STATE $\mathbf{L_{RF}} \leftarrow RF_s(\mathbf{L_{RF}},D_{freq}(v) ,M_{freq}(v) )$
        \STATE $\mathbf{P_{RF}} \leftarrow \frac{e^{\mathbf{L_{RF}}}}{\text{sum}(e^{\mathbf{L_{RF}}})}$
        \
        \STATE
$\text{$\cdot\cdot\cdot$Stego sampling}$
\STATE$x_{t}=w_{j} \leftarrow{f}{(B_i, \mathcal{K}, \mathbf{P_{RF}}  )}$
        \STATE $t \leftarrow t + 1$
    \ENDWHILE
    \RETURN   output stego text ${\bf x}$
\end{algorithmic}
\end{algorithm}

\subsection{Sequential Dimension}
\label{sequential dimension}
Our adjustments in the sequential dimension serve two purposes: first, to increase entropy while maintaining generation quality to enhance embedding capacity; and second, to counteract the model's sharpening effect \cite{lebrun2022evaluating}, a tendency where the model becomes overly confident in its predictions, which can reduce the diversity and creativity of the generated stego text.

We dynamically adjust the temperature at each time step $t$ based on the instantaneous entropy $E_t$ of the model:
\begin{equation}
   E_t = -\sum_{k \in \mathscr{V}} p(w_k | {\bf x}_{0:t-1}) \cdot \log[p(w_k | {\bf x}_{0:t-1})].
   \label{entro}
\end{equation}
By conditioning the temperature on the instantaneous entropy of each decoding step, FreStega avoids using a single fixed smoothing strength for all contexts.
\textcolor{black}{
This strategy enhances the model's ability to generate diverse stego text, better capturing the natural unpredictability and richness of human language.
The temperature $TP_t$ at time step $t$ is adjusted using the temperature function $Temp(\cdot)$:}
\begin{equation}{\color{black}
    TP_t=Temp(E_t)=1+\theta* \log_{2}(1+c*E_t)\label{temp},}
\end{equation}
\textcolor{black}{where $TP_t$ is the adjusted temperature at time step $t$, $c$ controls the adjustment strength, $\theta$ is a scaling factor that we typically fix at 0.01, and the base value 1 represents the default temperature typically used by Hugging Face LLMs, which can be determined by user settings.}
\textcolor{black}{
The logarithmic temperature has three useful properties. First, since $E_t\in[0,\log_2(|\mathscr{V}|)]$, $TP_t$ is bounded by $[1,1+\theta\log_2(1+c\log_2(|\mathscr{V}|))]$, keeping the adjustment mild. Second, the logarithm is increasing and concave: the temperature grows with the current entropy, but the marginal growth decreases, allowing entropy-aware exploration without excessive global over-smoothing. Low-entropy deterministic contexts therefore remain close to the base distribution. Third, dividing all logits by the same positive scalar preserves token ranking, so sequential adjustment changes distribution sharpness without replacing the model's preferred token with a semantically inconsistent one.}

The probability of each token is readjusted by the sequential probability distribution reform function $RF_t$.
\begin{equation}
\begin{split}
L_{RF_t}(x_t=w_k|{\bf x}_{0:t-1}) &=\frac{{\bf h}^{t} \cdot {\bf o}^{t}_{k}}{TP_t}. \\
&\label{divide t}
\end{split}
\end{equation}
$ L_{RF_t}(x_t=w_k|{\bf x}_{0:t-1}) $ is the logit of token \( w_k \) after sequential reformation $RF_t$ at time step $t$.
For numerical implementation, we apply a shared offset to these logits before spatial reweighting so that the resulting scores have a consistent sign. This shared shift preserves their order and is omitted from the notation for clarity.

\subsection{Spatial Dimension}
\label{spatial dimension}
Prior studies \cite{lebrun2022evaluating,TVD} have reported that LMs exhibit distorted preferences in next-token prediction during text generation. From a Bayesian perspective, the language distribution can be estimated as follows:

\begin{equation}
\begin{aligned}
& \hat{p}\left(x_{t}=w_{k}|{\bf x}_{0:t-1}\right)=\frac{p(\textbf{x}_{0:t-1}|x_t=w_k)p(x_t=w_k)}{p(\textbf{x}_{0:t-1})} \\
& =\begin{cases}\frac{p(\textbf{x}_{0:t-1}|x_t=w_k)}{p(\textbf{x}_{0:t-1})}\frac{n'_{k}}{\sum _{i}^{|\mathscr{V}|}{n'_i}} , \text{in\,generated\, text},
 \\\frac{p(\textbf{x}_{0:t-1}|x_t=w_k)}{p(\textbf{x}_{0:t-1})}\frac{n_{k}}{\sum _{i}^{|\mathscr{V}|}{n_i}} ,\text{ in\,training\, corpus}.
\end{cases}
\end{aligned}
\end{equation}
where $n_k$ and $n_i$ denote the frequencies of tokens $k$ and $i$ in the training corpus, while $n_k^{'}$ and $n_i^{'}$ denote their frequencies in generated text, respectively.

Tokens with the highest frequency in the training corpus appear even more frequently in the generated texts ($n_k^{'} > n_k$), while rare tokens become even rarer ($n_k^{'} < n_k$).

The intuitive idea behind the Spatial Adjustment function is to weaken the model's own bias towards certain tokens as much as possible and to closely approximate the distribution of tokens in the target cover corpus.

\textcolor{black}{For plug-and-play alignment with a target domain, we selected token frequency as our guiding signal because it represents the most intuitive and readily obtainable marginal distribution of text. From a statistical perspective, token frequency distributions in a corpus are widely regarded as empirical probability distributions of language, a view supported by both statistical natural language processing \cite{foundation} and information theory \cite{shannon1948mathematical}. Such distributions offer a clear and efficient proxy for capturing textual characteristics and can be seamlessly incorporated into the decoding process of a language model.}

To revise the biased LM distribution, we apply a spatial distribution reform function $RF_s$ to adjust the logits of a token $w_k$. The reformed logits $L_{RF_s}$ are shown in Equation \ref{Eq.7}:
\begin{equation}\label{Eq.7}
L_{RF_s}\left(w_{k}|{\bf x}_{0:t-1}\right) = \mathscr{F}(w_{k};\mathscr{D},\mathscr{M}) \cdot L_{RF_t}\left(\cdot\right),
\end{equation}

where $L_{RF_t}\left(\cdot\right)$ denotes the result of sequential reformation. $\mathscr{F}(w_{k};\mathscr{D},\mathscr{M})$ denotes the function computed from the unigram frequency of $w_{k}$ in the target corpus $\mathscr{D}$ and in the corpus generated by random sampling from the original model $\mathscr{M}$, as shown in Equation \ref{Eq.8}:
{
\begin{equation}\label{Eq.8}
    {\mathscr{F} (\textbf{w};\mathscr{D},\mathscr{M})=\log(2 +[\frac{f(\textbf{w}; \mathscr{D})}{f(\textbf{w}; \mathscr{M})}]^{\alpha}}),
\end{equation}
}

where $f(\textbf{w}; \mathscr{D})$ is the static unigram frequency of token $\textbf{w}$ in the target corpus $\mathscr{D}$, $f(\textbf{w}; \mathscr{M})$ is the static unigram frequency of token $\textbf{w}$ in the corpus generated by the original LM $\mathscr{M}$, and the
hyperparameter $\alpha$ controls the intensity of Spatial Adjustment.

The frequency adjustment function $\mathscr{F}(w_{k};\mathscr{D},\mathscr{M})$
shifts the resulting frequency distribution
closer to that of the target-domain cover corpus $\mathscr{D}$, diminishing the dominance of the frequency distribution inherent to the original model-generated text corpus $\mathscr{M}$.
	For $\alpha>0$, the function is positive and monotonically increasing with respect to the target-to-model frequency ratio $r$, so tokens that are more frequent in the target domain receive larger correction factors. Moreover, $\mathscr{F}(r)=O(\log r)$ and $\mathscr{F}(r)/r\to0$ as $r\to\infty$. Thus, extreme frequency ratios do not produce proportional increases in the correction factor. Unlike Sequential Adjustment, Spatial Adjustment may change token ranking because its purpose is to introduce a target-domain lexical preference; this sublinear growth moderates the influence of extreme ratios.
In this way, we encourage language models to output a distribution that more closely resembles target-domain cover texts. The reformed probability is
\begin{equation}\label{Eq.9}
p_{RF_s}\left(w_{k}|{\bf x}_{0:t-1}\right) = \frac{L_{RF_s}\left(w_{k}|{\bf x}_{0:t-1}\right)}{\sum_{w_{i}\in \mathscr{V}}L_{RF_s}\left(w_{i}|{\bf x}_{0:t-1}\right)}.
\end{equation}

\section{Experiments}

\subsection{Experiment Setup}

In our main experiment, we tested FreStega with four large language models, various steganography algorithms, and three datasets representing distinct human styles\footnote{\textcolor{black}{Our code and data are available at https://github.com/anonymousP0722/FreStega}}.

\subsubsection{Baselines}
FreStega requires no additional training modules such as classifiers. Existing attribute-controlled text generation methods are hard to apply directly to steganography and often need carefully designed training modules, making comparison with FreStega difficult.
We evaluated FreStega using the random sampling decoding method (viewed as the ideal steganography algorithm) and four state-of-the-art generative steganography algorithms: AC \cite{ziegler2019neural}, ADG \cite{zhang2021provably}, METEOR \cite{meteor2021}, and DISCOP \cite{ding2023discop}. These representative methods serve as baselines.

\subsubsection{Datasets}

We selected three datasets with distinct human characteristics that are challenging for models to mimic: IMDB \cite{IMDB}, XHS \cite{xhs}, and SHAKESPEARE\footnote{https://hf.com/datasets/ayaan04/shakespeare-text}. IMDB and SHAKESPEARE are English datasets, while XHS is in Chinese. IMDB and XHS represent popular online social platforms known for their informal, fragmented language, making them common scenarios for covert text communication. The SHAKESPEARE dataset is used to simulate situations where stego text needs to match the sender's usual style.
These datasets, due to their informality and low knowledge intensity, are typically underrepresented in LLM pre-training.
The corpus statistics are provided in the supplementary material.
In the main experiment, we used the entire target-domain corpus for alignment and explored the effect of varying the amount of target-domain data in subsequent analyses.

\subsubsection{Models}

We evaluate FreStega using four large language models (\textsc{\textcolor{black}{Mistral v0.3}} \cite{mistral}, \textsc{Qwen2} \cite{qwen}, \textsc{Llama2} \cite{llama2} and \textsc{Llama3} \cite{llama3}).
We use \textsc{\textcolor{black}{Mistral v0.3}} \cite{mistral}, \textsc{Qwen2} \cite{qwen}, \textsc{Llama2} \cite{llama2} and \textsc{Llama3} \cite{llama3} as the base models for the English datasets, with model scales of 7B, 7B, 7B, and 8B, respectively. Since the XHS dataset is entirely in Chinese, we use \textsc{Qwen2} 7B \cite{qwen} and \textcolor{black}{\textsc{ChatGLM3 6B}} \cite{du2022glm}, which were primarily trained on Chinese data.
Since large language models with around 7 billion parameters have already demonstrated strong instruction understanding capabilities, we conducted all our tests in a zero-shot scenario.
The prompts we used are detailed in the supplementary material. These prompts are loaded with the function {\texttt{apply\_chat\_template}} and then used as inputs to the LLMs.
For each test, we let the models generate 10,000 samples with top-$p=1.0$ and base temperature $1.0$.
All our experiments were conducted on 2 $\times$ NVIDIA A5000 GPUs (32GB RAM) and 24 $\times$ Intel Xeon w5-3423 CPUs.
\subsection{Metrics}
We primarily evaluate the quality of stego text based on three criteria: \textbf{linguistic quality} (including \textit{perplexity} and \textit{diversity}), \textbf{capacity} (including \textit{entropy} and \textit{embedding rate}), and \textbf{domain statistical imperceptibility} (including \textit{MAUVE} \cite{mauve} and \textit{steganalysis F1 scores}).

\subsubsection{Linguistic Quality}
\textcolor{black}{Although perplexity (PPL) does not fully capture the quality of generated text, it can partially reflect fluency.} A lower PPL indicates more fluent generated text. We use Qwen2-7B as the evaluator to calculate PPL in the experiments.
    $PPL = \exp\left(-\frac{1}{N}\sum_{i=1}^{N}\log p(x_i|{\bf x}_{1 : i - 1 })\right)$.

For text diversity, we use the $distinct_n$ metric, which computes the ratio of unique $n$-grams to all $n$-grams in the text.
$
    distinct_n = \frac{count(unique\, n-grams)}{count(n-grams)}.
$
\textcolor{black}{In the main experiments, we use} $distinct_3$.

\subsubsection{Capacity}

Entropy measures the richness of information in a sentence, with higher entropy indicating greater information uncertainty. It also represents the upper limit of generative linguistic steganographic capacity \cite{liao2024co}.
{
\begin{equation}
entropy = \sum_{i=1}^{N} \sum_{j=1}^{|\mathscr{V}|}p_i(w_j|{\bf x}_{1 : i - 1 }) \log p_i(w_j|{\bf x}_{1 : i - 1 }).
\end{equation}
}
\textcolor{black}{In our experiments, we use the average entropy per token as the measurement, defined as:}
$\textcolor{black}{
    \frac{entropy}{token} = \frac{1}{N}  \sum_{i=1}^{N} \sum_{j=1}^{|\mathscr{V}|}p_i(w_j|{\bf x}_{1 : i - 1 }) \log p_i(w_j|{\bf x}_{1 : i - 1 }).}
$

The Embedding Rate (ER) quantifies the payload of stego text by measuring the average number of embedded bits per token. A higher ER indicates a greater relative embedding capacity. The ER is calculated using the formula:
$
ER = \frac{L}{N_t},
$
where \(L\) is the length of the embedded message, and \(N_t\) is the number of tokens in the stego text.

\subsubsection{Domain Statistical Imperceptibility}

The main goal of steganography is to conceal a secret message so that it appears as an innocent carrier. \textcolor{black}{We mainly measure the imperceptibility of stego text from two perspectives: the degree of similarity to the distribution of real cover text in the target domain and the ability to resist detection by steganalysis. }

We use MAUVE \cite{mauve} to measure how closely the distribution of generated stego text matches that of target-domain human text (estimated by Monte Carlo methods). MAUVE ranges from $[0,1]$, and values closer to 1 indicate more similar distributions.

We also utilized three classic steganalysis models to assess the imperceptibility of stego texts: TS-CSW \cite{TS-CSW}, TS-RNN \cite{TS-RNN}, and R-BiLSTM-C \cite{niu2019hybrid}, based on pre-trained BERT \cite{BERT}.
We trained these detectors with 1,000 samples of generated stego texts and the target-domain human corpus in each test.
The datasets for training, validation, and testing were split in a ratio of 3:1:1. The learning rate was set to 1e-4, and training was conducted for three epochs.
After repeating this process three times, we evaluated the average F1 score of the test set as the steganalysis F1 score.

\begin{table*}
\centering
     \caption{\textcolor{black}{IMDB Main Result. "W/o" indicates models without FreStega, while "W/" indicates models with FreStega.}}
     \label{tab:Movie main result}
     \resizebox{\textwidth}{!}{
    \begin{tabular}{clccccccccc}
     \toprule[1.5pt]\hline
         \multirow{2}{*}{\textbf{Model}}&  \multirow{2}{*}{\textbf{Algorithm}}& \multicolumn{2}{c}{\textbf{Linguistic Quality}} & \multicolumn{2}{c}{\textbf{Capacity}}& \multicolumn{5}{c}{\textbf{Domain Statistical Imperceptibility}}\\
         \cmidrule(lr){3-4}
         \cmidrule(lr){5-6}
        \cmidrule(lr){7-11}

        ~ & ~ & \textbf{Div.}$\uparrow$ & \textbf{PPL}$\downarrow$ & \textbf{Entro./token}$\uparrow$ & \textbf{ER} $\uparrow$ & \textbf{MAU.}$\uparrow$ & \textbf{CSW-F1}$\downarrow$ & \textbf{RNN-F1}$\downarrow$ & \textbf{R-BiLSTM-C-F1}$\downarrow$ &  \\ \hline
        \multirow{10}{*}{\textsc{Llama2 \cite{llama2}}} & RS w/o & 0.12 & 3.03 & 0.80 & / & 0.05 & 83.59\%$\pm$0.85\% & 84.81\%$\pm$2.71\% & 84.16\%$\pm$7.16\% &  \\
        ~ & AC \cite{ziegler2019neural} w/o & 0.13 & 3.09 & 0.81 & 0.55 & 1.05 & 83.68\%$\pm$1.43\% & 81.89\%$\pm$2.43\% & 88.48\%$\pm$1.96\% &  \\
        ~ & ADG \cite{zhang2021provably} w/o & 0.12 & 3.02 & 0.79 & 0.10 & 0.87 & 85.72\%$\pm$2.25\% & 80.95\%$\pm$3.10\% & 86.07\%$\pm$2.87\% &  \\
        ~ & METEOR \cite{meteor2021} w/o & 0.11 & 2.97 & 0.79 & 0.29 & 0.84 & 85.52\%$\pm$1.42\% & 79.66\%$\pm$0.73\% & 81.71\%$\pm$5.49\% &  \\
        ~ & DISCOP \cite{ding2023discop} w/o & 0.12 & 3.00 & 0.80 & 0.25 & 3.41 & 83.61\%$\pm$3.21\% & 81.44\%$\pm$4.00\% & 84.31\%$\pm$2.84\% &  \\ \cmidrule{2-10}
        ~ & RS w/ & \textbf{0.30} & 3.07 & \textbf{0.81} & / & \textbf{4.37} & \textbf{69.67\%$\pm$1.98\%} & \textbf{64.81\%$\pm$2.71\%} & \textbf{70.27\%$\pm$3.12\%} &  \\
        ~ & AC \cite{ziegler2019neural} w/ & \textbf{0.34} & 3.20 & \textbf{0.82} & \textbf{0.56} & \textbf{3.66} & \textbf{74.23\%$\pm$1.21\%} & \textbf{67.23\%$\pm$4.00\%} & \textbf{71.23\%$\pm$5.65\%} &  \\
        ~ & ADG \cite{zhang2021provably} w/ & \textbf{0.30} & \textbf{2.85} & \textbf{0.79} & \textbf{0.12} & \textbf{5.67} & \textbf{72.65\%$\pm$3.64\%} & \textbf{65.45\%$\pm$1.87\%} & \textbf{67.51\%$\pm$3.86\%} &  \\
        ~ & METEOR \cite{meteor2021} w/ & \textbf{0.31} & 3.05 & \textbf{0.81} & \textbf{0.32} & \textbf{7.65} & \textbf{71.34\%$\pm$4.21\%} & \textbf{64.16\%$\pm$3.61\%} & \textbf{68.91\%$\pm$2.34\%} &  \\
        ~ & DISCOP\cite{ding2023discop} w/ & \textbf{0.30} & \textbf{3.00} & \textbf{0.82} & \textbf{0.28} & \textbf{6.60} & \textbf{65.23\%$\pm$1.85\%} & \textbf{68.68\%$\pm$2.95\%} & \textbf{70.39\%$\pm$2.21\%} &  \\ \cmidrule{1-10}
        \multirow{10}{*}{\textsc{Llama3 \cite{llama3}}} & RS w/o & 0.13 & 3.80 & 0.80 & / & 0.65 & 88.10\%$\pm$3.35\% & 81.80\%$\pm$4.50\% & 90.80\%$\pm$4.03\% &  \\
        ~ & AC \cite{ziegler2019neural} w/o & 0.14 & 3.96 & 0.82 & 0.64 & 0.48 & 88.62\%$\pm$2.62\% & 84.92\%$\pm$1.70\% & 87.07\%$\pm$3.90\% &  \\
        ~ & ADG \cite{zhang2021provably} w/o & 0.12 & 3.76 & 0.80 & 0.19 & 0.60 & 89.29\%$\pm$3.33\% & 83.80\%$\pm$4.50\% & 87.88\%$\pm$1.96\% &  \\
        ~ & METEOR \cite{meteor2021} w/o & 0.13 & 3.79 & 0.81 & 0.36 & 0.73 & 91.31\%$\pm$1.53\% & 81.41\%$\pm$2.58\% & 90.40\%$\pm$2.25\% &  \\
        ~ & DISCOP \cite{ding2023discop} w/o & 0.13 & 3.88 & 0.80 & 0.31 & 0.66 & 88.07\%$\pm$2.64\% & 82.33\%$\pm$3.05\% & 87.80\%$\pm$6.54\% &  \\ \cmidrule{2-10}
        ~ & RS w/ & \textbf{0.34} & \textbf{3.78} & \textbf{0.81} & / & \textbf{5.00} & \textbf{69.53\%$\pm$3.42\%} & \textbf{67.80\%$\pm$4.50\%} & \textbf{78.21\%$\pm$2.45\%} &  \\
        ~ & AC \cite{ziegler2019neural} w/ & \textbf{0.36} & 4.03 & \textbf{0.84} & \textbf{0.69} & \textbf{6.54} & \textbf{68.25\%$\pm$3.12\%} & \textbf{69.23\%$\pm$3.56\%} & \textbf{79.22\%$\pm$3.91\%} &  \\
        ~ & ADG \cite{zhang2021provably} w/ & \textbf{0.39} & 3.78 & \textbf{0.83} & \textbf{0.22} & \textbf{5.00} & \textbf{72.13\%$\pm$1.94\%} & \textbf{69.35\%$\pm$3.87\%} & \textbf{77.33\%$\pm$3.62\%} &  \\
        ~ & METEOR \cite{meteor2021} w/ & \textbf{0.41} & 4.13 & \textbf{0.84} & \textbf{0.41} & \textbf{8.25} & \textbf{66.49\%$\pm$2.13\%} & \textbf{68.31\%$\pm$1.43\%} & \textbf{70.92\%$\pm$1.49\%} &  \\
        ~ & DISCOP \cite{ding2023discop} w/ & \textbf{0.39} & \textbf{3.88} & \textbf{0.80} & \textbf{0.34} & \textbf{8.35} & \textbf{67.14\%$\pm$3.94\%} & \textbf{68.53\%$\pm$1.82\%} & \textbf{71.23\%$\pm$1.72\%} &  \\ \cmidrule{1-10}
        \multirow{10}{*}{\textsc{\textcolor{black}{Mistral v0.3} \cite{mistral}}} & RS w/o & 0.19 & 4.03 & 0.90 & / & 0.59 & 86.43\%$\pm$5.32\% & 80.07\%$\pm$5.13\% & 89.56\%$\pm$3.51\% &  \\
        ~ & AC \cite{ziegler2019neural} w/o & 0.18 & 3.91 & 0.88 & 0.95 & 0.55 & 87.03\%$\pm$2.02\% & 81.06\%$\pm$1.12\% & 90.52\%$\pm$3.26\% &  \\
        ~ & ADG \cite{zhang2021provably} w/o & 0.22 & 4.16 & 0.93 & 0.44 & 1.12 & 86.95\%$\pm$1.78\% & 88.91\%$\pm$3.27\% & 92.43\%$\pm$3.15\% &  \\
        ~ & METEOR \cite{meteor2021} w/o & 0.19 & 4.01 & 0.90 & 0.60 & 0.54 & 87.15\%$\pm$2.14\% & 83.88\%$\pm$3.99\% & 91.07\%$\pm$3.68\% &  \\
        ~ & DISCOP \cite{ding2023discop} w/o & 0.19 & 4.03 & 0.91 & 0.49 & 0.56 & 86.95\%$\pm$2.57\% & 85.77\%$\pm$1.45\% & 85.42\%$\pm$4.31\% &  \\ \cmidrule{2-10}
        ~ & RS w/ & \textbf{0.44} & \textbf{3.96} & 0.89 & / & \textbf{3.00} & \textbf{75.82\%$\pm$1.32\%} & \textbf{66.66\%$\pm$0.73\%} & \textbf{76.23\%$\pm$3.59\%} &  \\
        ~ & AC \cite{ziegler2019neural} w/ & \textbf{0.41} & \textbf{3.75} & \textbf{0.88} & \textbf{0.97} & \textbf{5.26} & \textbf{72.03\%$\pm$1.66\%} & \textbf{68.46\%$\pm$5.42\%} & \textbf{76.90\%$\pm$1.08\%} &  \\
        ~ & ADG \cite{zhang2021provably} w/ & \textbf{0.44} & 4.34 & \textbf{0.94} & \textbf{0.48} & \textbf{4.17} & \textbf{74.39\%$\pm$3.45\%} & \textbf{70.06\%$\pm$1.12\%} & \textbf{79.55\%$\pm$1.13\%} &  \\
        ~ & METEOR \cite{meteor2021} w/ & \textbf{0.48} & 4.29 & \textbf{0.92} & \textbf{0.70} & \textbf{4.22} & \textbf{74.13\%$\pm$2.89\%} & \textbf{68.95\%$\pm$1.28\%} & \textbf{74.46\%$\pm$3.59\%} &  \\
        ~ & DISCOP \cite{ding2023discop} w/ & \textbf{0.44} & 4.17 & \textbf{0.91} & \textbf{0.50} & \textbf{2.91} & \textbf{75.46\%$\pm$1.38\%} & \textbf{66.67\%$\pm$1.98\%} & \textbf{70.89\%$\pm$2.26\%} &  \\ \cmidrule{1-10}
        \multirow{10}{*}{\textsc{Qwen2 \cite{qwen}}} & RS w/o & 0.12 & 2.35 & 0.73 & / & 0.53 & 85.88\%$\pm$1.13\% & 81.75\%$\pm$3.23\% & 84.61\%$\pm$3.68\% &  \\
        ~ & AC \cite{ziegler2019neural} w/o & 0.14 & 2.53 & 0.74 & 0.65 & 0.66 & 85.36\%$\pm$3.39\% & 84.74\%$\pm$2.45\% & 80.92\%$\pm$4.26\% &  \\
        ~ & ADG \cite{zhang2021provably} w/o & 0.13 & 2.39 & 0.74 & 0.20 & 0.69 & 85.90\%$\pm$2.94\% & 88.07\%$\pm$2.12\% & 87.39\%$\pm$1.10\% &  \\
        ~ & METEOR \cite{meteor2021} w/o & 0.13 & 2.36 & 0.73 & 0.36 & 0.57 & 85.05\%$\pm$4.64\% & 80.61\%$\pm$2.85\% & 83.91\%$\pm$5.77\% &  \\
        ~ & DISCOP \cite{ding2023discop} w/o & 0.12 & 2.34 & 0.73 & 0.31 & 0.53 & 84.89\%$\pm$4.32\% & 83.01\%$\pm$3.01\% & 83.08\%$\pm$5.49\% &  \\ \cmidrule{2-10}
        ~ & RS w/ & \textbf{0.38} & 2.45 & \textbf{0.75} & / & \textbf{2.79} & \textbf{71.38\%$\pm$1.29\%} & \textbf{70.78\%$\pm$2.12\%} & \textbf{63.30\%$\pm$3.49\%} &  \\
        ~ & AC \cite{ziegler2019neural} w/ & \textbf{0.37} & \textbf{2.45} & \textbf{0.75} & \textbf{0.70} & \textbf{1.90} & \textbf{72.82\%$\pm$1.07\%} & \textbf{72.61\%$\pm$0.88\%} & \textbf{67.67\%$\pm$2.82\%} &  \\
        ~ & ADG \cite{zhang2021provably} w/ & \textbf{0.37} & \textbf{2.39} & \textbf{0.75} & \textbf{0.23} & \textbf{2.60} & \textbf{71.38\%$\pm$1.09\%} & \textbf{71.27\%$\pm$1.56\%} & \textbf{72.13\%$\pm$3.98\%} &  \\
        ~ & METEOR \cite{meteor2021} w/ & \textbf{0.38} & 2.43 & \textbf{0.74} & \textbf{0.41} & \textbf{2.38} & \textbf{71.12\%$\pm$3.38\%} & \textbf{60.56\%$\pm$2.19\%} & \textbf{65.92\%$\pm$1.18\%} &  \\
        ~ & DISCOP \cite{ding2023discop} w/ & \textbf{0.33} & 2.42 & \textbf{0.76} & \textbf{0.32} & \textbf{1.76} & \textbf{72.13\%$\pm$1.12\%} & \textbf{70.27\%$\pm$1.28\%} & \textbf{69.28\%$\pm$3.19\%} &  \\ \hline\bottomrule[1.5pt]
    \end{tabular}}
\end{table*}

\begin{table*}
     \caption{\textcolor{black}{XHS Main Result. "W/o" indicates models without FreStega, while "W/" indicates models with FreStega.}}
   \resizebox{\textwidth}{!}{
    \begin{tabular}{clcccccccc}
     \toprule[1.5pt]\hline
         \multirow{2}{*}{\textbf{Model}}&  \multirow{2}{*}{\textbf{Algorithm}}& \multicolumn{2}{c}{\textbf{Linguistic Quality}} & \multicolumn{2}{c}{\textbf{Capacity}}& \multicolumn{4}{c}{\textbf{Domain Statistical Imperceptibility}}\\
         \cmidrule(lr){3-4}
         \cmidrule(lr){5-6}
        \cmidrule(lr){7-10}

         & & \textbf{Div.}$\uparrow$& \textbf{PPL}$\downarrow$  &\textbf{Entro./token}$\uparrow$ & \textbf{ER} $\uparrow$ &\textbf{MAU.}$\uparrow$&\textbf{CSW-F1}$\downarrow$ &\textbf{RNN-F1}$\downarrow$ &\textbf{R-BiLSTM-C-F1}$\downarrow$\\\hline
 \multirow{10}{*}{\textsc{Qwen2\cite{qwen}}}& RS w/o &0.74& 15.83&2.22& /& 73.18&89.08\%$\pm$2.64\% & 85.18\%$\pm$0.98\%&79.73\%$\pm$2.64\%\\
 & AC\cite{ziegler2019neural} w/o & 0.76& 11.08&2.08& 2.83& 70.14& 90.41\%$\pm$2.62\%& 79.94\%$\pm$3.26\%&80.24\%$\pm$1.48\%\\
 & ADG\cite{zhang2021provably} w/o & 0.76& 21.35&2.06& 1.93& 63.54& 94.94\%$\pm$1.16\%& 74.19\%$\pm$3.60\%&83.34\%$\pm$2.10\%\\
 & METEOR\cite{meteor2021} w/o & 0.73& 16.19&2.18& 1.48& 73.69& 89.15\%$\pm$1.29\%& 80.26\%$\pm$1.95\%&85.82\%$\pm$2.93\%\\
 & DISCOP\cite{ding2023discop} w/o & 0.78& 11.02&2.13& 1.13 & 74.63& 87.81\%$\pm$2.15\%& 86.53\%$\pm$4.56\%&87.22\%$\pm$3.77\%\\
 \cmidrule{2-10}
& RS w/ &\textbf{0.80}& \textbf{15.77}&\textbf{2.30}& /& \textbf{91.05}&\textbf{57.79\%$\pm$1.81\%}&\textbf{47.05\%$\pm$2.63\%}&\textbf{66.43\%$\pm$2.94\%}\\
 & AC\cite{ziegler2019neural} w/ &\textbf{0.78}& 11.11&\textbf{2.17}& \textbf{2.99}& \textbf{93.97}& \textbf{52.85\%$\pm$3.24\%}&\textbf{50.05\%$\pm$2.74\%}&\textbf{71.28\%$\pm$1.38\%}\\
 &ADG\cite{zhang2021provably} w/ &\textbf{0.79}& \textbf{19.17}&\textbf{2.31}& \textbf{2.43}& \textbf{91.72}& \textbf{60.79\%$\pm$1.41\%}& \textbf{52.77\%$\pm$1.65\%}&\textbf{73.31\%$\pm$3.36\%}\\
 & METEOR\cite{meteor2021} w/ &\textbf{0.78}& \textbf{12.35}&\textbf{2.18}& \textbf{2.15}& \textbf{93.20}& \textbf{53.61\%$\pm$2.71\%}& \textbf{47.05\%$\pm$2.74\%}&\textbf{69.21\%$\pm$1.43\%}\\
 &DISCOP\cite{ding2023discop} w/ &\textbf{0.81}& 13.61&\textbf{2.17}& \textbf{1.16}& \textbf{92.47}&\textbf{55.56\%$\pm$2.90\%}&\textbf{45.77\%$\pm$2.55\%}&\textbf{65.29\%$\pm$1.16\%}\\\hline
 \multirow{10}{*}{\textsc{Chatglm3\cite{du2022glm}}}& RS w/o &0.75& 319.62 &4.55& /& 26.44& 62.19\%$\pm$4.50\%& 68.23\%$\pm$4.99\%&76.01\%$\pm$1.17\%\\
 & AC\cite{ziegler2019neural} w/o & 0.75& 308.82  &4.56& 4.23& 27.87& 68.51\%$\pm$3.37\%& 77.55\%$\pm$3.04\% &74.44\%$\pm$3.43\%\\
 & ADG\cite{zhang2021provably} w/o & 0.77& 394.76&4.61&4.36 & 54.43& 79.95\%$\pm$0.79\%&72.73\%$\pm$1.42\% &78.43\%$\pm$1.30\%\\
 & METEOR\cite{meteor2021} w/o & 0.70& 284.99  &4.08& 2.27& 31.35& 69.02\%$\pm$3.81\%& 73.82\%$\pm$4.56\%&64.52\%$\pm$3.11\%\\
 & DISCOP\cite{ding2023discop} w/o & 0.75& 311.86  &4.52& 1.29& 27.78& 65.80\%$\pm$4.59\% & 74.62\%$\pm$3.77\%&75.44\%$\pm$3.43\%\\
  \cmidrule{2-10}
 & RS w/ &\textbf{0.78}& 366.74 &\textbf{4.67}&  /& \textbf{80.17}&\textbf{44.44\%$\pm$1.43\%}&\textbf{43.44\%$\pm$3.21\%}&\textbf{34.23\%$\pm$2.92\%}\\
 & AC\cite{ziegler2019neural} w/ &\textbf{0.77}&366.31  &\textbf{4.66}& \textbf{4.53}& \textbf{93.21}& \textbf{47.21\%$\pm$1.51\%}& \textbf{45.77\%$\pm$2.99\%}&\textbf{38.59\%$\pm$2.31\%}\\
 & ADG\cite{zhang2021provably} w/ &\textbf{0.79}&\textbf{394.53}  &\textbf{4.88}&\textbf{4.82} &\textbf{71.11}&\textbf{52.33\%$\pm$1.02\%} & \textbf{52.26\%$\pm$1.42\%}&\textbf{54.78\%$\pm$1.29\%}\\
 & METEOR\cite{meteor2021} w/ &\textbf{0.78}&365.81 &\textbf{4.66}&\textbf{3.44} & \textbf{90.61}&\textbf{44.16\%$\pm$3.26\%}&\textbf{43.85\%$\pm$1.01\%}&\textbf{41.26\%$\pm$2.35\%}\\
 \rowcolor{gray} & DISCOP\cite{ding2023discop} w/ &\textbf{0.77}& 363.70 &\textbf{4.62}& \textbf{1.39}& \textbf{90.02}& \textbf{45.77\%$\pm$2.55\%}& \textbf{44.15\%$\pm$2.33\%}&\textbf{43.58\%$\pm$2.39\%}\\
 \hline\bottomrule[1.5pt]
    \end{tabular}
   }
    \label{tab:XHS main result}
\end{table*}

\begin{table*}
     \centering
          \caption{\textcolor{black}{SHAKESPEARE Main Result. "W/o" indicates models without FreStega, while "W/" indicates models with FreStega.}}
     \resizebox{\textwidth}{!}{
    \begin{tabular}{clcccccccc}
     \toprule[1.5pt]\hline
         \multirow{2}{*}{\textbf{Model}}&  \multirow{2}{*}{\textbf{Algorithm}}& \multicolumn{2}{c}{\textbf{Linguistic Quality}} & \multicolumn{2}{c}{\textbf{Capacity}}& \multicolumn{4}{c}{\textbf{Domain Statistical Imperceptibility}}\\
         \cmidrule(lr){3-4}
         \cmidrule(lr){5-6}
        \cmidrule(lr){7-10}

         & & \textbf{Div.}$\uparrow$& \textbf{PPL}$\downarrow$ &\textbf{Entro./token}$\uparrow$ & \textbf{ER} $\uparrow$ &\textbf{MAU.}$\uparrow$&\textbf{CSW-F1}$\downarrow$ &\textbf{RNN-F1}$\downarrow$ &\textbf{R-BiLSTM-C-F1}$\downarrow$\\\hline

 \multirow{10}{*}{\textsc{Llama2\cite{llama2}}}& RS w/o &0.31& 4.17&0.83& /& 9.12& 93.39\%$\pm$1.02\%& 94.81\%$\pm$0.93\%&96.52\%$\pm$0.50\%\\
 & AC\cite{ziegler2019neural} w/o & 0.30& 4.13 &0.83& 0.32& 9.09& 94.29\%$\pm$1.05\%& 95.21\%$\pm$0.92\%&96.31\%$\pm$0.32\%\\
 & ADG\cite{zhang2021provably} w/o & 0.29& 4.09 &0.80& 0.27&8.76& 94.45\%$\pm$0.25\%& 96.93\%$\pm$0.77\%&96.53\%$\pm$0.63\%\\
 & METEOR\cite{meteor2021} w/o & 0.29& 4.10&0.80& 0.28& 9.32& 93.25\%$\pm$1.20\%& 96.22\%$\pm$0.87\%&96.00\%$\pm$0.73\%\\
 & DISCOP\cite{ding2023discop} w/o & 0.29& 4.08&0.81& 0.14& 8.08& 95.25\%$\pm$1.40\%& 97.16\%$\pm$0.07\%&96.99\%$\pm$0.48\%\\\cmidrule{2-10}
         &  RS w/ &\textbf{0.45}& 4.38&\textbf{0.83}&  /& \textbf{19.52}& \textbf{69.96\%$\pm$2.26\%}& \textbf{70.13\%$\pm$2.27\%}&\textbf{70.81\%$\pm$7.49\%}\\
         &  AC\cite{ziegler2019neural} w/ &\textbf{0.47}& 4.58&\textbf{0.85}&  \textbf{0.38}& \textbf{20.15}& \textbf{69.56\%$\pm$1.12\%}& \textbf{69.98\%$\pm$1.37\%}&\textbf{72.17\%$\pm$4.56\%}\\
         &  ADG\cite{zhang2021provably} w/ &\textbf{0.46}&  4.42&\textbf{0.83}&  \textbf{0.35}& \textbf{22.47}& \textbf{67.01\%$\pm$1.82\%}& \textbf{66.62\%$\pm$2.90\%}&\textbf{68.45\%$\pm$2.59\%}\\
         &  METEOR\cite{meteor2021} w/ &\textbf{0.44}& 4.32&\textbf{0.84}& \textbf{0.31}&\textbf{21.99}& \textbf{68.96\%$\pm$1.12\%}& \textbf{68.80\%$\pm$1.83\%}&\textbf{70.89\%$\pm$5.60\%}\\
         &  DISCOP\cite{ding2023discop} w/ &\textbf{0.43}&  4.29&\textbf{0.83}&  \textbf{0.17}& \textbf{24.64}& \textbf{67.52\%$\pm$1.81\%}& \textbf{66.62\%$\pm$2.90\%}&\textbf{70.46\%$\pm$3.12\%}\\\cmidrule{1-10}
  \multirow{10}{*}{\textsc{Llama3\cite{llama3}}}& RS w/o &0.37& 3.86&0.70& /& 0.54& 83.70\%$\pm$0.48\%& 87.64\%$\pm$1.34\%&86.00\%$\pm$0.29\%\\
 & AC\cite{ziegler2019neural} w/o & 0.35& 3.83&0.70& 0.36& 0.54& 84.92\%$\pm$0.32\%& 89.14\%$\pm$2.46\%&87.39\%$\pm$1.69\%\\
 & ADG\cite{zhang2021provably} w/o & 0.36& 3.84&0.71& 0.33& 0.49& 81.02\%$\pm$4.88\%& 86.35\%$\pm$1.51\%&85.47\%$\pm$1.95\%\\
 & METEOR\cite{meteor2021} w/o & 0.37& 3.85&0.71& 0.23& 0.56& 82.27\%$\pm$2.12\%& 87.23\%$\pm$0.49\%&84.78\%$\pm$1.40\%\\
 & DISCOP\cite{ding2023discop} w/o & 0.35& 3.83&0.71&0.17& 0.56& 83.13\%$\pm$0.78\%& 86.95\%$\pm$0.59\%&85.21\%$\pm$0.99\%\\\cmidrule{2-10}
 & RS w/ &\textbf{0.47}& \textbf{3.82}&\textbf{0.70}& /& \textbf{2.08}& \textbf{60.91\%$\pm$2.25\%}& \textbf{58.62\%$\pm$3.92\%}&\textbf{62.03\%$\pm$3.80\%}\\
 & AC\cite{ziegler2019neural} w/ &\textbf{0.48}& 3.93&\textbf{0.71}& \textbf{0.39}& \textbf{2.03}& \textbf{61.27\%$\pm$1.32\%}& \textbf{59.74\%$\pm$1.69\%}&\textbf{65.09\%$\pm$3.35\%}\\
 & ADG\cite{zhang2021provably} w/ &\textbf{0.49}&3.94&\textbf{0.73}& \textbf{0.37}& \textbf{3.14}& \textbf{61.29\%$\pm$1.07\%}& \textbf{57.41\%$\pm$1.31\%}&\textbf{64.33\%$\pm$2.08\%}\\
 & METEOR\cite{meteor2021} w/ &\textbf{0.49}&3.96&\textbf{0.71}& \textbf{0.26}& \textbf{1.95}& \textbf{59.56\%$\pm$4.52\%}& \textbf{54.72\%$\pm$1.67\%}&\textbf{65.48\%$\pm$3.70\%}\\
 & DISCOP\cite{ding2023discop} w/ &\textbf{0.50}& 3.99&\textbf{0.71}& \textbf{0.18}& \textbf{2.60}& \textbf{58.87\%$\pm$3.08\%}& \textbf{57.41\%$\pm$4.31\%}&\textbf{65.71\%$\pm$2.58\%}\\\hline
 \multirow{10}{*}{\textsc{\textcolor{black}{Mistral v0.3}\cite{mistral}}}& RS w/o &0.34& 4.02&0.72& /& 1.04& 84.86\%$\pm$0.82\%& 88.37\%$\pm$0.71\%&85.60\%$\pm$1.18\%\\
 & AC\cite{ziegler2019neural} w/o & 0.37& 4.08 &0.72& 0.23& 1.18& 86.31\%$\pm$1.95\%& 90.36\%$\pm$2.47\%&86.73\%$\pm$2.04\%\\
 & ADG\cite{zhang2021provably} w/o & 0.37& 4.07&0.71& 0.22& 1.14& 83.17\%$\pm$0.80\%& 88.55\%$\pm$0.57\%&86.95\%$\pm$1.45\%\\
 & METEOR\cite{meteor2021} w/o & 0.36& 4.06&0.72&0.16& 1.16& 83.55\%$\pm$1.45\%& 88.04\%$\pm$0.86\%&86.86\%$\pm$0.34\%\\
 & DISCOP\cite{ding2023discop} w/o & 0.35&4.05&0.72& 0.11& 1.29& 83.59\%$\pm$1.29\%& 87.58\%$\pm$0.64\%&86.52\%$\pm$0.30\%\\\cmidrule{2-10}
 & RS w/ &\textbf{0.49}&\textbf{4.02}&\textbf{0.73}& /& \textbf{35.75}& \textbf{59.67\%$\pm$2.78\%}& \textbf{65.57\%$\pm$1.87\%}&\textbf{66.21\%$\pm$3.60\%}\\
 & AC\cite{ziegler2019neural} w/ &\textbf{0.52}& 4.11&\textbf{0.73}& \textbf{0.26}& \textbf{29.73}& \textbf{59.32\%$\pm$3.02\%}& \textbf{64.82\%$\pm$2.30\%}&\textbf{66.27\%$\pm$3.29\%}\\
 & ADG\cite{zhang2021provably} w/ &\textbf{0.49}&\textbf{4.00}&\textbf{0.72}& \textbf{0.24}& \textbf{12.10}& \textbf{59.22\%$\pm$3.58\%}& \textbf{63.27\%$\pm$4.86\%}&\textbf{67.24\%$\pm$2.43\%}\\
 & METEOR\cite{meteor2021} w/ &\textbf{0.48}&4.18&\textbf{0.74}& \textbf{0.17}& \textbf{23.45}& \textbf{57.13\%$\pm$1.75\%}& \textbf{64.84\%$\pm$2.30\%}&\textbf{66.30\%$\pm$4.14\%}\\
 & DISCOP\cite{ding2023discop} w/ &\textbf{0.52}&4.25&\textbf{0.73}& \textbf{0.12}& \textbf{37.56}& \textbf{58.32\%$\pm$1.86\%}& \textbf{63.96\%$\pm$2.27\%}&\textbf{66.32\%$\pm$3.16\%}\\\hline
 \multirow{10}{*}{\textsc{Qwen2\cite{qwen}}}& RS w/o &0.39& 3.06&0.72& /& 4.98& 82.10\%$\pm$1.15\%& 81.17\%$\pm$0.36\%&81.06\%$\pm$1.76\%\\
 & AC\cite{ziegler2019neural} w/o & 0.40& 3.05&0.72& 0.46& 4.00& 83.97\%$\pm$2.71\%& 82.88\%$\pm$2.06\%&82.15\%$\pm$1.56\%\\
 & ADG\cite{zhang2021provably} w/o & 0.41& 3.08&0.73& 0.43& 4.18& 81.93\%$\pm$0.74\%& 81.38\%$\pm$1.02\%&80.36\%$\pm$0.92\%\\
 & METEOR\cite{meteor2021} w/o & 0.40& 3.07&0.73&0.31& 4.27& 80.68\%$\pm$1.45\%& 68.61\%$\pm$2.85\%&80.26\%$\pm$1.28\%\\
 & DISCOP\cite{ding2023discop} w/o & 0.40& 3.07&0.73& 0.20& 4.96& 80.84\%$\pm$1.00\%& 81.17\%$\pm$2.01\%&81.54\%$\pm$1.20\%\\\cmidrule{2-10}
 & RS w/ &\textbf{0.51}& \textbf{3.04}&\textbf{0.74}& /& \textbf{9.04}& \textbf{58.46\%$\pm$2.87\%}& \textbf{52.63\%$\pm$1.68\%}&\textbf{62.36\%$\pm$3.83\%}\\
 & AC\cite{ziegler2019neural} w/ &\textbf{0.53}& 3.19&\textbf{0.76}& \textbf{0.49}& \textbf{8.80}& \textbf{57.18\%$\pm$2.58\%}& \textbf{52.21\%$\pm$4.68\%}&\textbf{62.34\%$\pm$3.92\%}\\
 & ADG\cite{zhang2021provably} w/ &\textbf{0.51}& \textbf{3.03}&\textbf{0.75}& \textbf{0.44}& \textbf{14.21}& \textbf{54.92\%$\pm$2.21\%}& \textbf{54.41\%$\pm$4.67\%}&\textbf{64.26\%$\pm$3.73\%}\\
 & METEOR\cite{meteor2021} w/ &\textbf{0.53}&3.20&\textbf{0.74}& \textbf{0.33}& \textbf{14.29}& \textbf{55.58\%$\pm$2.86\%}& \textbf{53.98\%$\pm$3.14\%}&\textbf{61.38\%$\pm$3.65\%}\\
 & DISCOP\cite{ding2023discop} w/ &\textbf{0.52}&3.18&\textbf{0.75}& \textbf{0.21}& \textbf{8.19}& \textbf{58.60\%$\pm$2.77\%}& \textbf{52.63\%$\pm$3.63\%}&\textbf{62.40\%$\pm$3.57\%}\\
 \hline\bottomrule[1.5pt]
    \end{tabular} }
    \label{tab:shakespere main result}
\end{table*}

\subsection{Performance Evaluation}

The experimental results of FreStega on the IMDB, XHS, and SHAKESPEARE datasets are shown in Tables~\ref{tab:Movie main result}, \ref{tab:XHS main result}, and~\ref{tab:shakespere main result}, respectively.
Note that the term ``RS'' refers to random sampling in language models without embedding any hidden messages, representing an ideal secure steganography baseline. In the main experiment, we set $c=0.1$ and $\alpha=0.1$. From these results, we have the following findings.
\subsubsection{Imperceptibility}
Even with advanced distribution-preserving methods like METEOR \cite{meteor2021} or DISCOP \cite{ding2023discop}, stego texts remain distinguishable from cover texts in real-world scenarios, as reflected by both MAUVE and steganalysis F1.
FreStega significantly improves imperceptibility across most scenarios, as evidenced by increased MAUVE scores and decreased detector F1. It better aligns the stego text distribution with human text in the target domain. We analyzed the word frequencies of stego texts generated using the baseline and FreStega. We observed that our method introduced some words that had not appeared before and that reflect distinct stylistic characteristics of the target corpus.
For example, words such as ``tempest'' and ``alas'', together with other Shakespearean expressions, appear after alignment. On the XHS platform, the language model learns to emulate the corresponding expressive style. Emojis appear more often in aligned XHS stego text.

\subsubsection{Linguistic Quality}

In terms of linguistic quality, it is evident that the diversity of stego texts generated by baseline methods is significantly lower than that of human text, particularly in the IMDB and SHAKESPEARE datasets, as reflected by $distinct_3$.
 A major reason existing baseline methods are easily detected is that their language model distributions are too sharp, causing the generated stego text to fall into a few fixed patterns. In contrast, FreStega adjusts the language model distribution to alleviate the sharpening phenomenon, effectively increasing the diversity of stego text in almost all scenarios.
Similarly, our method ensures the fluency of the generated text, with the perplexity being almost identical to that of baseline methods. \textcolor{black}{In most cases, since FreStega mitigates the distribution sharpness issue in language models, incorporating FreStega leads to a slight increase in perplexity, typically by no more than 5\%—except for ChatGLM3-6B. In a few cases, PPL may even decrease, but overall it remains roughly comparable to the baseline without FreStega.
}
Overall, FreStega effectively preserves the meaningfulness and coherence of the stego text, even as it enhances the entropy in the sequential alignment stage.

\subsubsection{Capacity}
In terms of capacity, FreStega consistently and reliably increases the embedding rate across all scenarios. FreStega achieved relative increases in embedding rate of 20.65\%, 10.07\%, and 15.94\% on the XHS, IMDB, and SHAKESPEARE datasets, respectively. This aligns with FreStega's intended goal of using sequential alignment to raise the entropy of the language model's distribution while maintaining text quality, thereby increasing the capacity upper bound over the original distribution. Experiments confirm that FreStega effectively increases the embedding rate while preserving text quality.

\begin{figure}[!t]
    \centering
    \includegraphics[width=1\linewidth]{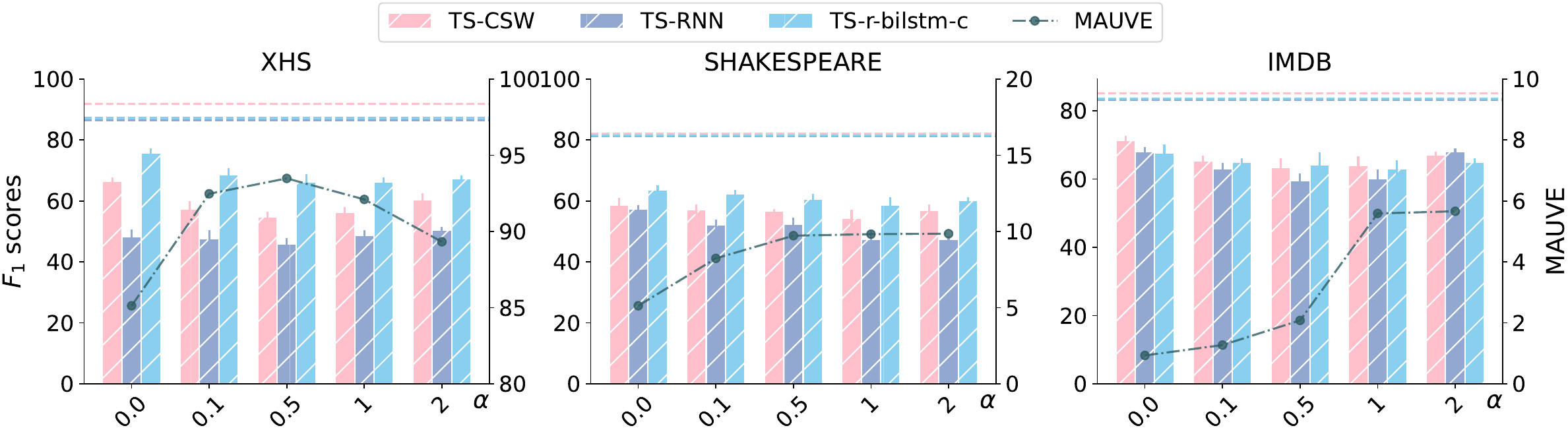}
    \caption{Hyperparameter analysis of $\alpha$ ($c$=0.1). We tested the F1 scores of classifiers on three datasets (XHS, SHAKESPEARE, IMDB) using three classic steganalysis methods (TS-CSW \cite{TS-CSW}, TS-RNN \cite{TS-RNN}, R-BiLSTM-C \cite{r-bilstm-c}). The corresponding colored horizontal lines indicate the F1 scores of the steganalysis methods without FreStega. The green dashed line represents the MAUVE score relative to the target-domain text.}
    \label{fig:alpha}
\end{figure}

\begin{figure}[!t]
    \centering
    \includegraphics[width=1\linewidth]{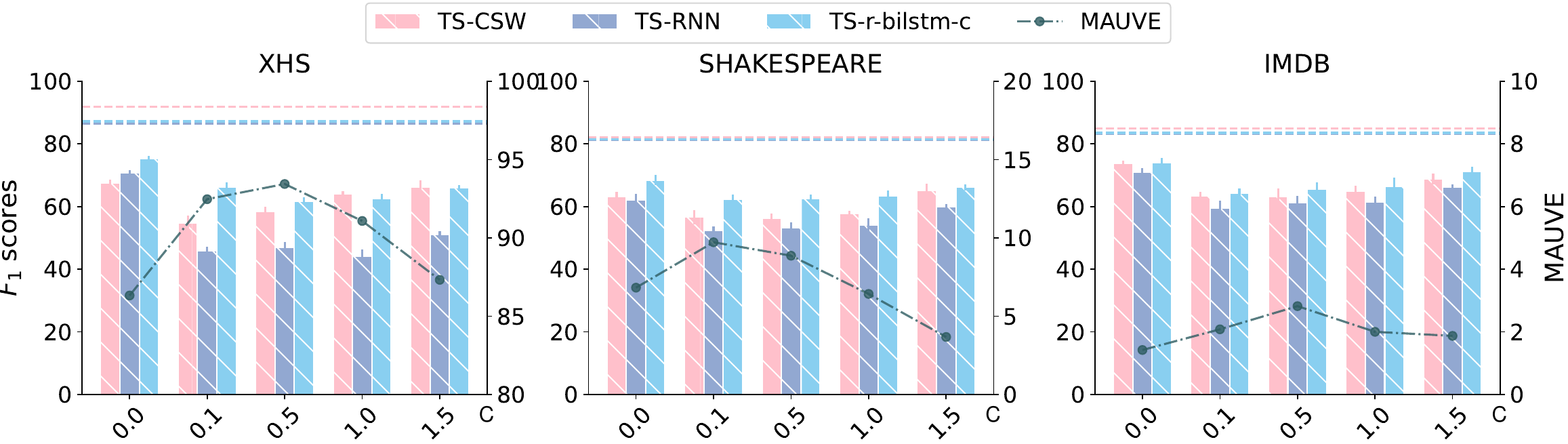}
    \caption{Hyperparameter analysis of $c$ ($\alpha$=0.1). We tested the F1 scores of classifiers on three datasets (XHS, SHAKESPEARE, IMDB) using three classic steganalysis methods (TS-CSW \cite{TS-CSW}, TS-RNN \cite{TS-RNN}, R-BiLSTM-C \cite{r-bilstm-c}). The corresponding colored horizontal lines indicate the F1 scores of the steganalysis methods without FreStega. The green dashed line represents the MAUVE score relative to the target-domain text.}
    \label{fig:c}
\end{figure}

\subsection{Control Intensity of Alignment}
\label{intensity}
For FreStega, we explored two key hyperparameters—spatial alignment intensity $\alpha$ and sequential alignment intensity $c$—using DISCOP on \textsc{Qwen2} as the representative steganography scheme with a base temperature of 1.

As Figure \ref{fig:alpha} shows, increasing $\alpha$ leads to better alignment with the target human text distribution, reflected in higher MAUVE values, while maintaining statistical imperceptibility. Although higher $\alpha$ enhances imperceptibility, the benefits stabilize, making extensive tuning unnecessary; empirically, $\alpha=0.1$ generally works well. The sequential alignment intensity \(c\) slightly adjusts the temperature based on the distribution's entropy, offering a global adjustment. While higher temperatures can increase text diversity, they may also add unnecessary noise. As shown in Figure \ref{fig:c}, we empirically set \(c = 0.1\) to improve embedding capacity while maintaining imperceptibility. \textcolor{black}{While not optimal, these hyperparameters perform well in most cases without further tuning.}

\begin{figure*}[!t]
    \centering
    \includegraphics[width=\textwidth]{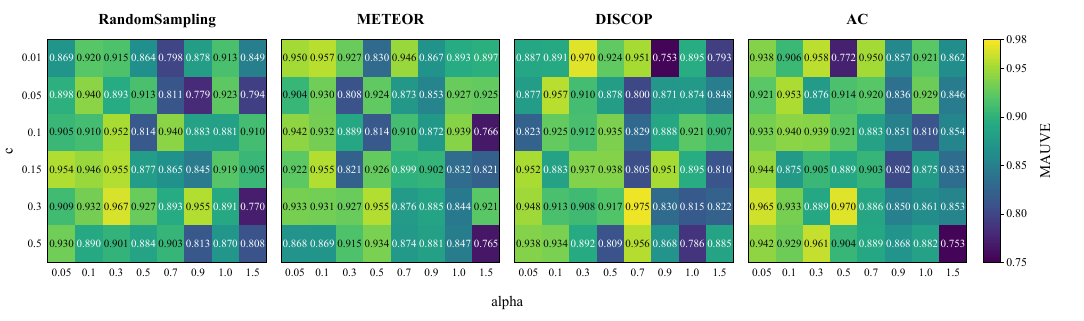}
   \caption{\textcolor{black}{Heatmaps showing MAUVE scores under different $(\alpha, c)$ settings for RandomSampling, METEOR, DISCOP, and AC on XHS.}}
    \label{fig:four-figs}
\end{figure*}

\textcolor{black}{
In practice, automatic hyperparameter optimization methods such as grid search can be employed. Specifically, we first generate small batches of stego text using FreStega on a small dataset (e.g., 500 samples) and then perform grid search over the hyperparameters $\alpha$ and $c$ on this subset to identify optimal configurations. The grid search results are presented in Figure~\ref{fig:four-figs}. FreStega achieves robust alignment with the target domain across a broad range of hyperparameter values (with MAUVE scores exceeding 0.8 in most cases). Although our default settings ($c = 0.1$, $\alpha = 0.1$) do not yield the absolute highest MAUVE scores, they consistently deliver strong performance across nearly all conditions.
}

\subsection{Quantity of Target Domain Texts Used for Alignment}

In the main experiment, we used all available target-domain text for alignment. In subsequent analyses, we discuss in detail the impact of the amount of target-domain data required for effective alignment.
We experimented with different quantities of target-domain text for alignment on the XHS dataset using two models, \textsc{Qwen2} and \textsc{ChatGLM3}, to assess the impact on stego text generation. The XHS dataset contains a total of 1,508 texts. To control for generation-size effects, we generated 2,000 texts for each setting. We then used 25, 50, 100, 300, 500, 700, 1,000, 1,300, and 1,500 target-domain texts for spatial alignment. The results are shown in Figures~\ref{fig:qwen-length} and~\ref{chatglm-length}.
The figures show that more alignment samples lead to better domain imperceptibility and higher fluency. FreStega requires only a small amount of target-domain text, around 100 samples, to effectively exhibit the characteristics of the target domain and closely match its distribution.
In practical deployment, if the domain distribution shifts over time, FreStega does not require model retraining; the token-frequency table can be incrementally updated by tokenizing newly collected samples and merging their counts with the existing statistics.
Even without access to target-domain data, the Sequential Adjustment of FreStega remains effective (the ablation results are shown in the supplementary material) and alleviates the inherent sharpening tendency of the language model. This fallback is particularly useful when the target corpus, such as XHS, exhibits high diversity.

\begin{figure*}[!t]
    \centering
    \includegraphics[width=1\linewidth]{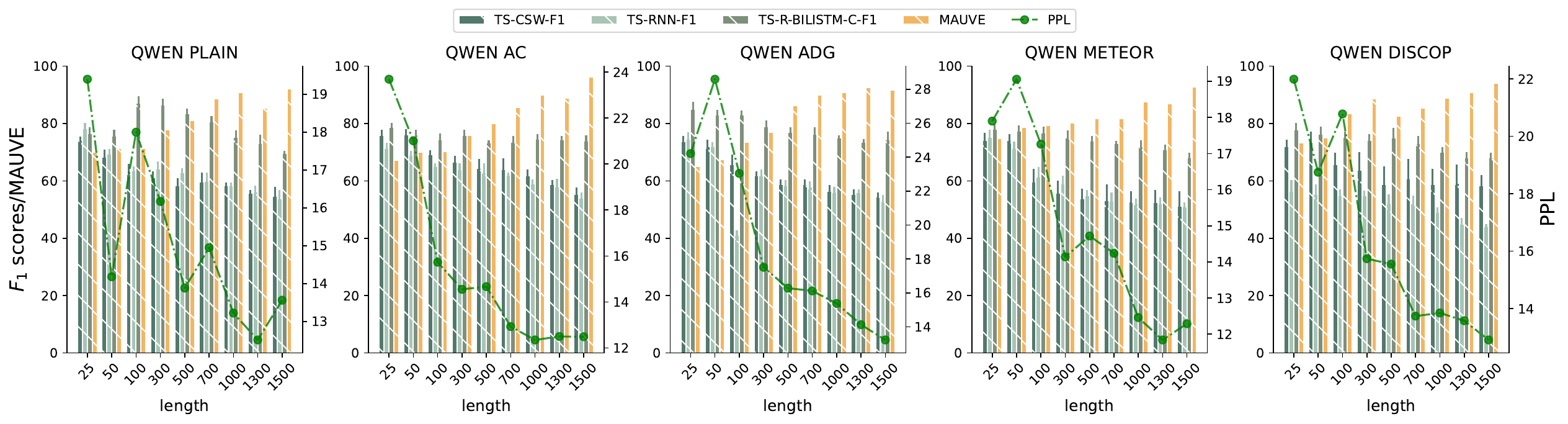}
    \caption{The effect of the number of target-domain texts used for alignment on \textsc{Qwen2} for XHS.}
    \label{fig:qwen-length}
\end{figure*}
\begin{figure*}[!t]
    \centering
    \includegraphics[width=1\linewidth]{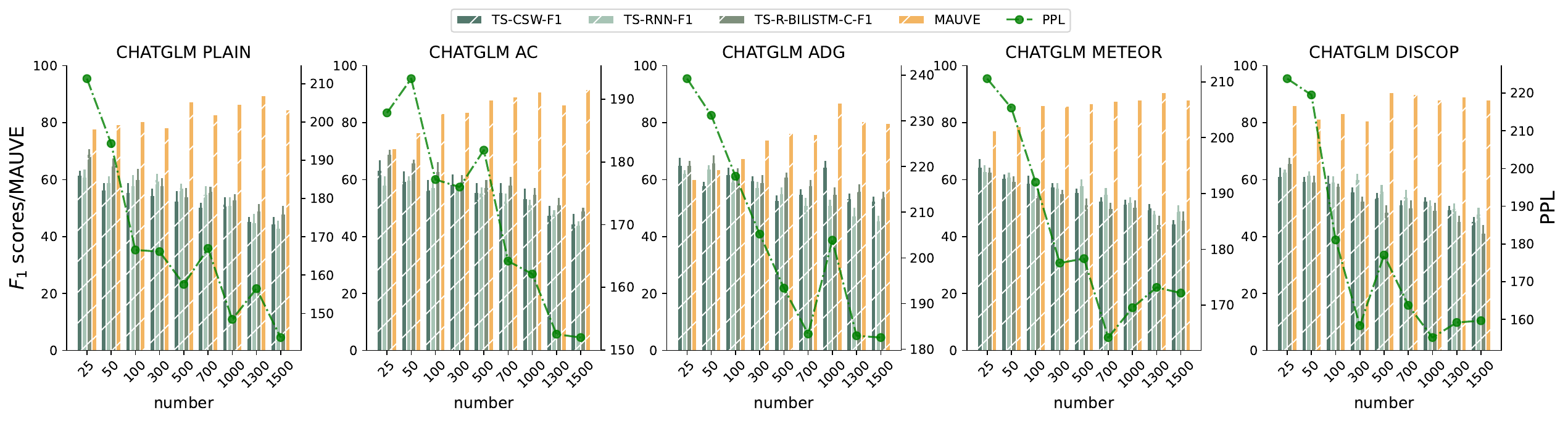}
    \caption{The effect of the number of target-domain texts used for alignment on \textsc{ChatGLM3} for XHS.}
    \label{chatglm-length}
\end{figure*}

\subsection{Applicability to LLMs of Different Sizes}
\label{modelsize}
FreStega is adaptable across model scales. In the main experiments, we primarily employed models around the 7B scale. To further evaluate scalability, we tested FreStega on the \textsc{Qwen2} series with different parameter sizes using zero-shot prompts on the XHS dataset, measuring MAUVE scores and embedding rates over 1,000 generated samples. As shown in Figure~\ref{fig:size}, FreStega remains effective across language models of different sizes while enhancing both imperceptibility and embedding capacity. The embedding rate drops when moving from smaller to larger models, reflecting the tendency of larger models to generate more deterministic distributions; this reinforces the importance of Sequential Adjustment for practical steganographic capacity.

\begin{figure}[!b]
    \centering
    \includegraphics[width=0.95\columnwidth]{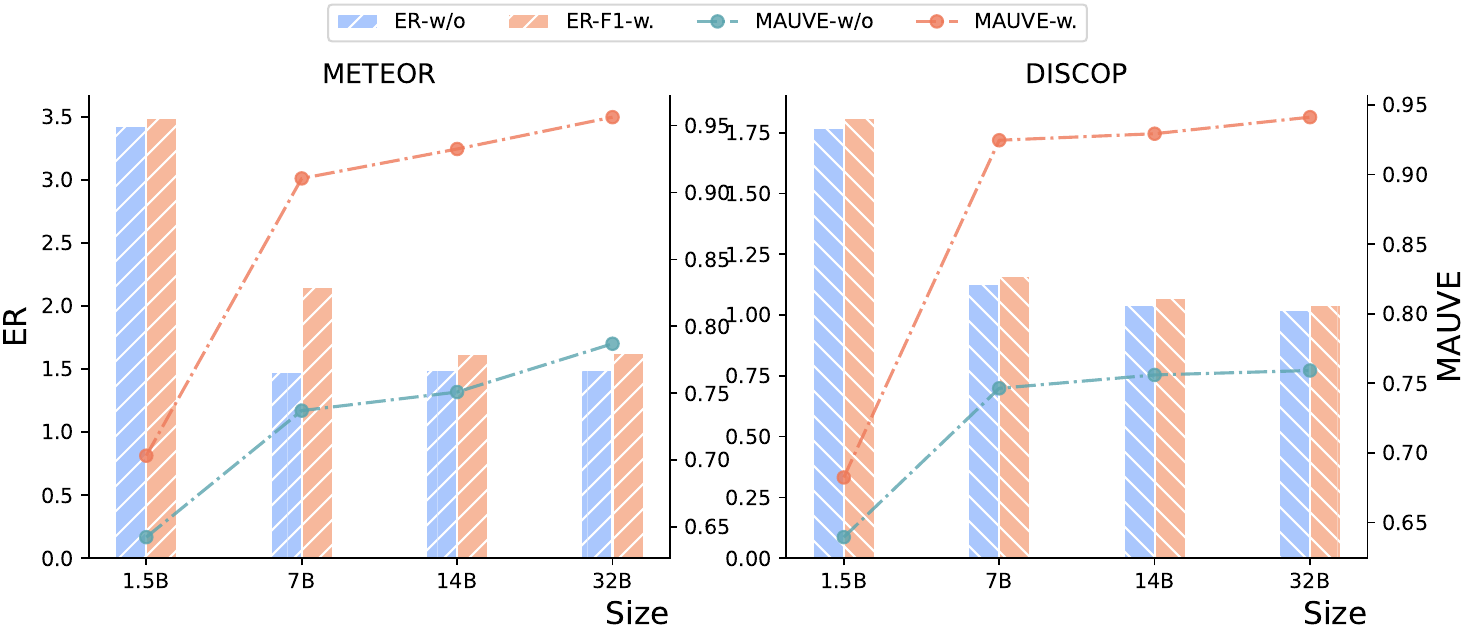}
    \caption{Results for stego texts generated by \textsc{Qwen2} models of different sizes.}
    \label{fig:size}
\end{figure}

\subsection{Time Efficiency and Dynamic Updates}
\label{timeefficiency}
We analyze time efficiency from two perspectives: online generation latency and offline preprocessing/update cost. Online generation latency reflects the runtime overhead during stego text generation, while offline preprocessing measures the one-time cost of building and refreshing the frequency statistics used by Spatial Adjustment.

\subsubsection{Online Generation Time}
FreStega only modifies token probabilities before sampling, so its online cost is limited. As shown in Table~\ref{time efficiency}, the ms/token values with and without FreStega are close across models and datasets. Paired \(t\)-tests further show that all \(p\)-values are above 0.05, indicating no statistically significant change in generation latency.

\begin{table*}[t]
    \centering
      \caption{Average generation latency (ms/token) for 1000 samples tested on different models across SHAKESPEARE, IMDB, and XHS.}
      \resizebox{\textwidth}{!}{
\begin{tabular}{lcccccccccccccccc}\toprule[1.5pt]\hline
\textbf{Dataset} & \multicolumn{5}{c}{\textbf{SHAKESPEARE}} & \multicolumn{5}{c}{\textbf{IMDB}} & \multicolumn{5}{c}{\textbf{XHS}} & \multirow{2}{*}{\textbf{p-value}}\\
  \cmidrule(lr){2-6} \cmidrule(lr){7-11} \cmidrule(lr){12-16}
         \textbf{Model} & \textbf{RS} & \textbf{AC\cite{ziegler2019neural}} & \textbf{ADG\cite{zhang2021provably}} & \textbf{Met.\cite{meteor2021}} & \textbf{Dis.\cite{ding2023discop}} & \textbf{RS} & \textbf{AC\cite{ziegler2019neural}} & \textbf{ADG\cite{zhang2021provably}} & \textbf{Met.\cite{meteor2021}} & \textbf{Dis.\cite{ding2023discop}} & \textbf{RS} & \textbf{AC\cite{ziegler2019neural}} & \textbf{ADG\cite{zhang2021provably}} & \textbf{Met.\cite{meteor2021}} & \textbf{Dis.\cite{ding2023discop}} \\\midrule
         \textsc{Llama2\cite{llama2}} w/o&  52.19&52.15 &61.13 & 52.43&56.72 & 58.54& 56.19& 59.79& 54.66&60.62 & -& -& -& -&-&\multirow{2}{*}{0.8640}\\
         \textsc{Llama2\cite{llama2}} w/& 52.07&53.19 & 62.27& 53.69&57.12 & 56.75& 55.85& 58.26& 55.72&60.12 &- &- &- &- &-&\\\midrule
 \textsc{Llama3\cite{llama3}} w/o& 63.17&61.19 & 102.99& 61.59&81.44 & 71.09& 72.47& 76.61& 62.74&81.56 & -&- &- &- &-&\multirow{2}{*}{0.5537}\\
 \textsc{Llama3\cite{llama3}} w/& 68.05&61.67 & 108.02& 64.58&88.42 & 63.74& 63.86& 76.87& 68.25&81.74 & -&- &- &- &-&\\\midrule
 \textsc{\textcolor{black}{Mistral v0.3}\cite{mistral}} w/o& 56.74&56.72 & 66.31& 56.93&65.29 & 58.13& 57.79& 63.45& 58.74&63.69 & -&- & -& -&-&\multirow{2}{*}{0.1351}\\
 \textsc{\textcolor{black}{Mistral v0.3}\cite{mistral}} w/& 57.85&56.64 & 66.13& 58.12&66.27& 58.43& 63.14& 63.41& 59.06&63.40 & -&- &- &- &-&\\\midrule
         \textsc{\textcolor{black}{Qwen2}\cite{qwen}} w/o& 58.49&57.74 & 106.05& 57.76&81.21 & 68.55& 66.01& 75.94& 60.47&84.67 & 58.72& 58.68& 150.28& 58.75&81.55&\multirow{2}{*}{0.3603}\\
           \textsc{\textcolor{black}{Qwen2}\cite{qwen}} w/& 62.27&58.35 & 109.24& 61.09&86.73 & 74.67& 66.89& 76.81& 62.31&84.59 & 60.84& 60.70& 120.23& 62.90&82.99&\\\midrule
 \textsc{\textcolor{black}{ChatGLM3}\cite{du2022glm}} w/o&- &- &- &- &- & -& -& -& -&- & 29.42& 29.76& 46.76& 30.71&38.87&\multirow{2}{*}{0.1199}\\
 \textsc{\textcolor{black}{ChatGLM3}\cite{du2022glm}} w/&- &- & -& -& -& -& -& -&- &- & 30.17& 30.40& 46.98& 30.43&39.35&\\
 \hline\bottomrule[1.5pt]
    \end{tabular}}
    \label{time efficiency}
\end{table*}

\subsubsection{Offline Preprocessing and Dynamic Updates}
\label{preprocessingcost}
Before Spatial Adjustment, FreStega builds token-frequency tables for the target domain and model-sampled texts. This is CPU tokenizer counting, not model training or detector inference. Table~\ref{preprocessing} reports representative costs with the \textsc{Qwen2} tokenizer. The table is reusable across generation runs, so the one-time cost is amortized. For domain drift, FreStega only tokenizes newly collected samples and merges their counts into the existing table; no retraining or full-corpus recomputation is required.

Due to space constraints, additional experiment details, compatibility analysis, an ablation study, and qualitative case studies are included in the supplementary material.

\begin{table}[!t]
    \centering
    \caption{Preprocessing time for target-domain frequency statistics.}
    \setlength{\tabcolsep}{4pt}
    \renewcommand{\arraystretch}{1.0}
    \begin{tabular}{@{}lrrr@{}}
    \toprule[1.5pt]\hline
    \textbf{Dataset} & \textbf{Entries} & \textbf{Tokens} & \textbf{Time (s)}\\
    \midrule
    XHS & 1,508 & 588,334 & 1.69\\
    IMDB & 1,039,403 & 25,809,893 & 73.04\\
    SHAKESPEARE & 18,395 & 241,212 & 1.04\\
    \hline\bottomrule[1.5pt]
    \end{tabular}
    \label{preprocessing}
\end{table}

\FloatBarrier

\newcommand{\FreStegaCaseStudy}{
\subsection{Case Study}
\label{casestudy}
Tables~\ref{casestudyIMDB}, \ref{casestudyShake}, and~\ref{casestudyXHS} present selected human references, stego texts generated without FreStega, and stego texts generated with FreStega for IMDB, SHAKESPEARE, and XHS, respectively.
The examples show that FreStega makes the generated stego texts more consistent with target-domain language habits, including colloquial IMDB phrasing, Shakespearean expressions, and emojis and tags in XHS. These qualitative results complement the quantitative imperceptibility evaluation by illustrating how distribution alignment changes the generated text in practical scenarios.

\begin{table*}[!t]
\begin{minipage}[t]{0.48\textwidth}
    \centering
    \caption{Examples in IMDB using \textsc{Qwen2}+DISCOP.}
    \label{casestudyIMDB}
    \footnotesize
    \setlength{\tabcolsep}{3pt}
    \renewcommand{\arraystretch}{0.98}
    \begin{tabularx}{\linewidth}{@{}>{\centering\arraybackslash}p{0.20\linewidth}X@{}}
        \toprule[1.5pt]
        \textbf{Method} & \multicolumn{1}{c}{\textbf{Generated Text}} \\
        \midrule
        Human & the movie was ok, but the special fx are cheesy and the singing and dancing is mediocre. \\\midrule
        w/o FreStega & A captivating, suspenseful film that masterfully combines thrilling action sequences with intricate plot twists and compelling characters. \\\midrule
        w/ FreStega & i liked it more than i expected, tho the ending was kinda dumb and the fx looked cheap.\\
        \bottomrule[1.5pt]
    \end{tabularx}
\end{minipage}\hfill
\begin{minipage}[t]{0.48\textwidth}
    \centering
    \caption{Examples in SHAKESPEARE using \textsc{Qwen2}+DISCOP.}
    \label{casestudyShake}
    \footnotesize
    \setlength{\tabcolsep}{3pt}
    \renewcommand{\arraystretch}{0.98}
    \begin{tabularx}{\linewidth}{@{}>{\centering\arraybackslash}p{0.20\linewidth}X@{}}
        \toprule[1.5pt]
        \textbf{Method} & \multicolumn{1}{c}{\textbf{Generated Text}} \\
        \midrule
        Prompt & You are a model that converts all of my input normal text into Shakespearean styled text. My input text is: I have half a mind to hit you before you speak again. \\\midrule
        Human & I have half a mind to strike thee ere thou speak'st again.\\\midrule
        w/o FreStega & Thy next word might provoke mine hand, anon. \\\midrule
        w/ FreStega & I am half resolved to strike thee ere thou utter'st another word.\\
        \bottomrule[1.5pt]
    \end{tabularx}
\end{minipage}
\vspace{1.5em}

    \centering
    \caption{Examples in XHS using \textsc{Qwen2}+DISCOP.}
    \label{casestudyXHS}
    \footnotesize
    \setlength{\tabcolsep}{3pt}
    \renewcommand{\arraystretch}{0.98}
    \begin{tabularx}{\textwidth}{@{}>{\centering\arraybackslash}p{0.12\textwidth}X@{}}
        \toprule[1.5pt]\hline
        \textbf{Method} & \multicolumn{1}{c}{\textbf{Generated Text}} \\
        \midrule
        Prompt & 请以小红书博主的口吻，以挑战15天马甲线第13天为标题写一篇小红书分享\\\midrule
        Human & 标题：\includegraphics[height=1em]{figures/emoji/power.png} 15天马甲线倒计时第13天 \includegraphics[height=1em]{figures/emoji/happy.png} 第13天打卡！今天照镜子，小腹还是有一点点，但腰两侧好像紧了一些[害羞R][害羞R] 最明显的变化不是体重，而是做核心时没以前那么抖了，卷腹也终于不是脖子在发力,发力点对了效果快\includegraphics[height=1em]{figures/emoji/happy.png} \includegraphics[height=1em]{figures/emoji/pray.png} 今天练了平板支撑、死虫、反向卷腹和俄罗斯转体\includegraphics[height=1em]{figures/emoji/dance.png} 后两组真的酸到想躺平，但还是坚持做完了！\includegraphics[height=1em]{figures/emoji/fire.png} \includegraphics[height=1em]{figures/emoji/ok.png} 饮食也没节食，正常吃饭，少喝奶茶少吃夜宵，早餐会吃鸡蛋和咖啡，午饭正常吃肉和菜，晚上稍微清淡一点。
偶尔嘴馋还是会想吃零食，但我现在不会因为吃了一点就自暴自弃了。吃了就吃了，下一顿继续正常吃，心态比以前稳很多。 \includegraphics[height=1em]{figures/emoji/salad.png}\includegraphics[height=1em]{figures/emoji/apple.png} 还有2天，希望能看到一点点线条！\includegraphics[height=1em]{figures/emoji/good.png} \#15天马甲线挑战 \#马甲线打卡 \#居家运动 \#核心训练 \#减脂日记
        \\\midrule
        w/o FreStega & 标题：\includegraphics[height=1em]{figures/emoji/power.png} 15天马甲线挑战：Days 1-13小记:坚持就是胜利！宝贝们！今天是马甲线挑战的第13天了！简直不敢相信时间过得这么快！(\#马甲线挑战 \#15天挑战）这13天中，我感到了自己的变化，从一开始的困惑和犹豫到现在的自信和坚持。\includegraphics[height=1em]{figures/emoji/fire.png}这距离第一次开始计时已经过了快两周，开始我只觉得这个计划好遥远。但现在，我每天醒来都充满期待，期待自己今天的进步和挑战更强大的自己。每一天的运动都是对自己的一次小挑战：腹部训练、瑜伽、深蹲、仰卧起坐……每一个动作虽然简单，但汇聚在一起就是改变。我开始留意肚皮在收缩时的感觉，那种奇妙的奇迹感让我每次都会露出满足的微笑。失砂の岁月，魔界丶鹤烽，而这不仅仅是一种外观的改变，它也是心态的转变。我知道，真正的改变是从心中开始的。每天看着镜子中的自己，我变得更加自信，挑战身体极限。\\\midrule
        w/ FreStega & 标题：\includegraphics[height=1em]{figures/emoji/power.png} 15天马甲线倒计时第13天：身体的小秘密，只对坚持的人展示！\includegraphics[height=1em]{figures/emoji/star.png}正文: 大家好！亲爱的小红书的小伙伴们\includegraphics[height=1em]{figures/emoji/hand.png}，我是你们的健身小伙伴，今天继续和大家分享我的马甲线挑战进度啦! \includegraphics[height=1em]{figures/emoji/fire.png} 第13天，我已经感觉自己和马甲线之间的距离一天比一天近! \includegraphics[height=1em]{figures/emoji/run.png}\includegraphics[height=1em]{figures/emoji/weiqi.png} 每次做起来腹部动作时，总能感觉到肌肉在暗示我：你越来越接近目标了呢~ \includegraphics[height=1em]{figures/emoji/happy.png}这13天的经历，让我深刻认识到坚持的力量！力量远远超出了我的想象。每一天的付出都在预扣着未来的好身材, 我得珍惜每一次的汗水\includegraphics[height=1em]{figures/emoji/power.png}。为了今天的效果，昨晚的晚餐我选择了清淡的沙拉搭配鸡胸肉，今天早餐也是高蛋白的蛋白饼，都是为了给肌肉更好的恢复和生长环境。\includegraphics[height=1em]{figures/emoji/salad.png}\includegraphics[height=1em]{figures/emoji/egg.png }营养永远是健身的伴侣，这13天我学会了更多关于饮食的知识！\includegraphics[height=1em]{figures/emoji/think.png}摄入要讲究，摄入要健康！明天就是这15天挑战的最后两天啦！激动又紧张，不过我知道成功要靠自己一步一步走过去，我不会放弃的！\includegraphics[height=1em]{figures/emoji/laba.png} \#练出腹肌马甲线\#瘦全身\\
        \bottomrule[1.5pt]
    \end{tabularx}
\end{table*}

}

\section{Conclusion}
We propose a plug-and-play probability distribution reconstruction method to enhance the imperceptibility and embedding capacity of generative linguistic steganography in real-world scenarios.
\textcolor{black}{State-of-the-art distribution-preserving steganography methods are closely tied to the language model’s native distribution, leading to a mismatch with real-world cover texts.
To address the resulting imperceptibility gap, we reconstruct the model distribution to align with the target domain across both sequential and spatial dimensions, enabling a more suitable and realistic channel for steganographic sampling.} FreStega seamlessly integrates with existing generative steganography algorithms and various language models.
Stego texts generated by FreStega demonstrate higher embedding capacity while maintaining linguistic quality, exhibiting strong imperceptibility against various steganalysis detectors in real-world scenarios.

\bibliographystyle{IEEEtran}
\FloatBarrier
{\footnotesize
\makeatletter
\let\FreStegaOriginalLBibitem\@lbibitem
\def\@lbibitem[#1]#2{
  \ifnum\c@NAT@ctr=37\relax\newpage\fi
  \FreStegaOriginalLBibitem[#1]{#2}}
\makeatother
\bibliography{sample,IEEEabrv}
}

\end{CJK*}
\end{document}